\documentclass[%
reprint,
superscriptaddress,
frontmatterverbose, 
preprintnumbers,
longbibliography,
amsmath,amssymb,
aps,
pra,
notitlepage,
nofootinbib,
twocolumn
]{revtex4-2}

\usepackage{lipsum}

\usepackage{graphicx}
\usepackage{dcolumn}
\usepackage{bm}
\usepackage{hyperref}
\usepackage{subfigure}
\hypersetup{
colorlinks=true,
linkcolor=blue,
filecolor=blue,
citecolor=blue,  
urlcolor=blue,
}

\newcommand{\mycomment}[1]{}

\usepackage[dvipsnames]{xcolor}

\usepackage{comment}

\usepackage{thmtools}
\usepackage{thm-restate}

\usepackage{enumerate} 

\usepackage{amssymb}
\usepackage{mathtools}

\usepackage{mathrsfs}
\usepackage{multirow}
\usepackage{bbm}

\usepackage[dvipsnames]{xcolor}

\definecolor{sanddune}{rgb}{0.59, 0.44, 0.09}
\definecolor{darkblue}{RGB}{0,0,102}
\definecolor{darkred}{rgb}{0.5,0.,0.}
\definecolor{BlueViolet}{RGB}{138,43,226}
\definecolor{SkyBlue}{RGB}{30,144,255}
\definecolor{DarkGreen}{RGB}{0,100,0}

\usepackage{amsthm}
\usepackage{amsmath}
\theoremstyle{plain}
\newtheorem{thm}{Theorem}
\newtheorem{lem}[thm]{Lemma}
\newtheorem{prop}[thm]{Proposition}
\newtheorem{cor}[thm]{Corollary}

\theoremstyle{definition}
\newtheorem{defn}{Definition}

\newtheorem{exmp}{Example}

\newtheorem{fact}{Fact}

\newtheorem*{rem}{Remark}

\newcommand{\ket}[1]{|#1\rangle}
\newcommand{\bra}[1]{\langle #1|}
\newcommand{\bracket}[2]{\langle #1|#2\rangle}
\newcommand{\ketbra}[2]{|#1\rangle\langle #2|}

\newcommand{\abs}[1]{\left|#1\right|}

\newcommand{\mc}{\mathcal}

\newcommand{\mbb}{\mathbb}

\newcommand{\stab}{\mathrm{STAB}}

\newcommand{\ba}{\begin{eqnarray}}
\newcommand{\ea}{\end{eqnarray}}

\DeclareMathOperator{\Tr}{Tr}

\newcommand{\mD}{\mathfrak{D}}

\newcommand{\diag}{\mathrm{diag}}

\newcommand{\lrom}{\mathrm{L}\mathcal{R}}

\newcommand{\zw}[1]{{#1}} 

\usepackage{makecell}
\usepackage{multirow}

\begin{document}
\newcommand{\onenorm}[1]{\left\| #1 \right\|_1}
\newcommand{\twonorm}[1]{\left\| #1 \right\|_2}

\newcommand{\ols}[1]{\mskip.5\thinmuskip\overline{\mskip-.5\thinmuskip {#1} \mskip-.5\thinmuskip}\mskip.5\thinmuskip} 
\newcommand{\olsi}[1]{\,\overline{\!{#1}}} 

\title{Entirely nonlocal quantum magic without entanglement}

\author{Fuchuan Wei}
\thanks{These authors contributed equally to this work.}
\affiliation{Yau Mathematical Sciences Center, Tsinghua University, Beijing 100084, China}

\author{Ruixia Wang}
\thanks{These authors contributed equally to this work.}
\affiliation{Beijing Key Laboratory of Fault-Tolerant Quantum Computing, Beijing Academy of Quantum Information Sciences, Beijing 100193, China}

\author{Yujia Zhang}
\affiliation{Beijing Key Laboratory of Fault-Tolerant Quantum Computing, Beijing Academy of Quantum Information Sciences, Beijing 100193, China}

\author{Huihui Li}
\affiliation{Yau Mathematical Sciences Center, Tsinghua University, Beijing 100084, China}

\author{Junfeng Li}
\affiliation{Yau Mathematical Sciences Center, Tsinghua University, Beijing 100084, China}

\author{Fei Yan}
\email{yanfei@baqis.ac.cn}
\affiliation{Beijing Key Laboratory of Fault-Tolerant Quantum Computing, Beijing Academy of Quantum Information Sciences, Beijing 100193, China}

\author{Zi-Wen Liu}
\email{zwliu0@tsinghua.edu.cn}
\affiliation{Yau Mathematical Sciences Center, Tsinghua University, Beijing 100084, China}

\date{\today}

\begin{abstract}
Nonstabilizerness, or magic, is an archetypal \emph{quantum} resource that is necessary for quantum computational advantage. Here we uncover a phenomenon seemingly at odds with the quantum nature of magic: entirely nonlocal magic (ENM)---magic present in joint correlations while absent locally---can exist
without entanglement.  We systematically study this separation and show that it is universal and operational: every magical state or channel can be encoded into and recovered from a separable ENM realization using only local stabilizer processing and classical communication. Building on this mechanism, we devise an activation key protocol that enables a provider to control user access to quantum computational power using a classical key. We further introduce magic secret sharing, a cooperative scheme that unlocks quantum computational power unavailable to any party
alone. On a superconducting quantum processor, we experimentally demonstrate activation key and network computing primitives, together with separable ENM state preparation and extraction protocols. Together, our results establish that magic can be classically activated, localized, and secret-shared without entanglement, providing new resource-control primitives for distributed quantum computation.
\end{abstract}

\maketitle


\section{Introduction}

\zw{A central theme of  quantum information is the separation
between global capability and local access: a joint quantum system can
enable information processing tasks that no constituent subsystem can
perform alone.
Entanglement, the paradigmatic nonclassical feature of composite quantum systems, provides the canonical quantum mechanism behind this separation.}
\zw{Entanglement enables teleportation and nonlocal gates and underpins distributed quantum computation and quantum secret sharing~\cite{
grover1997quantumtelecomputation,Cirac1999Distributed,
Hillery1999secret,Main2025Distributed}.
Quantum data hiding provides another prominent manifestation of the same principle~\cite{Terhal2001Hiding}.
These paradigms naturally suggest an essential role of shared entanglement in distributing access to quantum resources across
separated parties.}

\zw{
Because quantum computation promises capabilities beyond classical reach, identifying and controlling the resources that enable this advantage is a central goal of quantum information science.
This has spurred extensive study of nonstabilizerness, or ``magic,'' a resource feature necessary for quantum computational advantage~\cite{
gottesman1998heisenberg,Aaronson2004Improved,
Bravyi2005universal,Veitch2014resource,Campbell2017Roads}. 
Beyond its computational role, magic has emerged as an important probe of many-body structure, dynamics, and critical phenomena, revealing aspects of quantum complexity not captured by entanglement~\cite{
White2021cft,Ellison2021symmetryprotected,Lorenzo2022stabilizer,
Liu2022manybody,Lami2023Nonstabilizerness,Tarabunga2023MBM,
Haug2023Scalable,Frau2024Nonstabilizerness,Tarabunga2024MPS,
Bejan2024Dynamical,Qian2024Augmenting,Niroula2024Phase,
Paviglianiti2025Estimating,Gu2025Separation,
wei2025longrangenonstabilizernessquantumcodes,Dowling2025Bridging,
Ding2025Evaluating,Falcao2025MBL,Turkeshi2025spreading,
Hoshino2026sreCFT,cao2026suddendeathentanglementrebirth,
dutta2026magicsecretsharingthreshold}.}

\zw{It is then natural to ask how magic is distributed and accessed across multipartite quantum systems, a question with direct relevance to quantum networking and distributed quantum computation.
In this work, we uncover a phenomenon that sharply challenges the presumed essential role of entanglement in quantumness distribution: magic can live solely in correlations  that are free of entanglement.}

More specifically, referring to globally present but locally absent magic as \emph{entirely nonlocal magic} (ENM), here we study ENM without entanglement (illustration in Fig.~\ref{fig:top}). In particular, we present systematic constructions of separable ENM and harness them to control when and where hidden magic becomes accessible.
We first show that any magical target state can be hidden in a separable ENM state and recovered exactly at a designated party: a classical label records a suitable Clifford randomization that makes both local marginals magic-free while retaining the correction needed for recovery.
Symmetric constructions remove this preferred direction and allow the target to be recovered at either party, albeit probabilistically.
For pure single-qubit targets, we further determine the optimal one-round recovery probability in the symmetric two-qubit setting.
We also identify a separable two-qubit ``golden'' ENM state that attains the maximum possible robustness of magic among all two-qubit states.
The same principle extends from states to dynamics, allowing any magical channel to be hidden in a separable ENM channel with magic-free marginal dynamics and recovered exactly using stabilizer operations and classical communication.
Leveraging these observations, we design a \emph{magic activation key} that controls access to quantum computational power: withholding the key leaves the effective dynamics magic-free, whereas revealing it restores the intended operation.
Building on the same principle, we develop \emph{magic secret sharing}, in which cooperative measurements and classical communication make globally encoded magic accessible at a node that cannot access it alone.

\zw{To substantiate our theoretical findings, we experimentally demonstrate representative ENM constructions and applications on a superconducting quantum processor, including representative ENM state preparation and extraction, and our magic activation key and secret sharing protocols.}

\begin{figure}[t]
\centering
\includegraphics[width=0.45\textwidth]{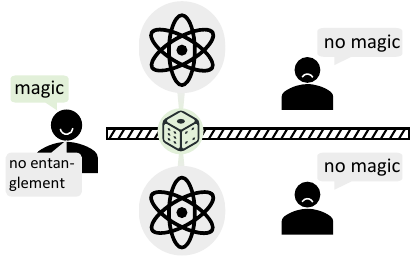}
\caption{Entirely nonlocal magic without entanglement.
In a separable bipartite state or channel, magic can be encoded entirely in correlations while both local reductions remain magic-free.
Cooperation using stabilizer operations and classical communication recovers the hidden resource.}
\label{fig:top}
\end{figure}

\section{Entirely nonlocal magic without entanglement}

\zw{We begin by developing a systematic framework for constructing separable ENM states and channels and recovering their encoded magic, thereby establishing  ENM as an operational resource.
We call the recovery one-way when it is directed deterministically to a fixed party, and two-way when either party may be selected as the recipient.
The framework is universal for arbitrary magical states and channels, while concrete realizations range from simple $T$-state encodings to an extremal separable ENM state attaining the maximum possible two-qubit magic.
All encoding and recovery procedures use only stabilizer processing and classical communication and require neither shared entanglement nor auxiliary magic.}

Formal definitions of ENM and separability, together with the stabilizer preliminaries, are provided in Methods.

\subsection{One-way separable ENM states}

\zw{Classical--quantum correlations provide a systematic route to one-way separable ENM states.
A suitable Clifford randomization, recorded in a classical label, makes both local marginals stabilizer while retaining the information needed to recover the target exactly.
The following theorem formalizes this reversible construction; see Appendix~\ref{app:proof_of_one_way} for the proof.}

\begin{thm}[One-way separable ENM states]\label{thm:separableENM}
Let $\sigma\notin\stab_n$.
Suppose a Clifford ensemble $\{p_i,W_i\}_{i=0}^{M-1}$ satisfies $\sum_ip_i W_i\sigma W_i^\dagger\in\stab_n$.
Then there exists a $(\lceil\log_2M\rceil+n)$-qubit separable ENM state $\rho_\sigma^{\rightarrow}$ and stabilizer protocols $\mc{E}_1$, $\mc{E}_2$, such that $\mc{E}_1(\sigma)=\rho_\sigma^{\rightarrow}$ and $\mc{E}_2(\rho_\sigma^{\rightarrow})=\sigma$.
\end{thm}

\zw{An explicit encoding scheme is
\begin{equation}
\rho_\sigma^{\rightarrow}
\coloneqq
\sum\nolimits_i p_i\ketbra{i}{i}_R\otimes W_i\sigma W_i^\dagger,
\label{eq:rho-def-main}
\end{equation}
where $R$ is an $m$-qubit classical label register with
$m=\lceil\log_2M\rceil$.
The construction is universal: for every nonstabilizer $\sigma$, the uniform ensemble of all $n$-qubit Pauli $Z$ strings completely dephases $\sigma$ in the computational basis, so the averaged state lies in $\stab_n$.
The following elementary example illustrates the resulting encoding and recovery:}

\begin{exmp}
\label{exmp:one-way-T-states}
Let $\ket{T^\perp}=Z\ket{T}$.
The state
\begin{equation}
\rho_T^{\rightarrow}\coloneqq \frac{1}{2}\ketbra{0}{0}\otimes\ketbra{T}{T}+\frac{1}{2}\ketbra{1}{1}\otimes\ketbra{T^\perp}{T^\perp},
\end{equation}
is a separable ENM state.
\end{exmp}

Measuring the first qubit and communicating the outcome $b\in\{0,1\}$ allows the second party to apply $Z^b$ and deterministically recover $\ket{T}$.
\zw{A single classical label bit can similarly mask two copies of the $T$ state:}  $\rho_{TT}^{\rightarrow}=\frac{1}{2}\sum_{b=0}^{1}\ketbra{b}{b}\otimes(Z^b\ketbra{T}{T}Z^b)^{\otimes 2}$. Once $b$ is revealed, $(Z^b)^{\otimes 2}$ recovers $\ket{T}^{\otimes 2}$.
More generally, reversibility under stabilizer protocols allows families of one-way separable ENM states to retain exponentially large magic; see Methods.
Because the key register stores only classical information, it need not be maintained in quantum memory, making $\rho_T^{\rightarrow}$ more robust under noise than its entangled analogue; see Appendix~\ref{app:Oneway_T_comparasion}.

\subsection{Two-way separable ENM states}

\zw{The deterministic one-way protocol designates a fixed recipient.
This directional restriction can nevertheless be removed for any target state: symmetrizing Eq.~\eqref{eq:rho-def-main} and adding flag qubits yields an embedding from which either party can recover $\sigma$ with probability $1/2$; see Appendix~\ref{app:flagged-two-way-embedding}.

A more physically transparent alternative dispenses with auxiliary flags altogether.
The state is formed by mixing the target branch $\sigma\otimes\sigma$ with a correlated stabilizer component chosen to neutralize magic in each local marginal while retaining globally recoverable magic in the joint state.
If the stabilizer component can be filtered out by a suitable local measurement, the successful outcome enables exact recovery of $\sigma$ on the other subsystem:}

\begin{thm}[Two-way separable ENM states]\label{thm:two_way_general}
Let $\sigma\notin\stab_n$.
If there exist $\tau\in\stab_n$, $p\in(0,1]$, and a stabilizer measurement $\Pi$ satisfying $p\sigma+(1-p)\tau\in\stab_n$ and $\Tr(\Pi\tau)=0<\Tr(\Pi\sigma)$, then
\begin{equation}
\rho^{\leftrightarrow}_{\sigma}\coloneqq p\,\sigma\otimes\sigma+(1-p)\,\tau\otimes\tau\label{eq:general-two-way-enm}
\end{equation}
is separable and ENM.
Local stabilizer operations and classical communication can extract $\sigma$ on either subsystem with probability $p_{\mathrm{succ}}=p\Tr(\Pi\sigma)$.
\end{thm}

\begin{figure*}[t]
\centering
\includegraphics[width=0.97\textwidth]{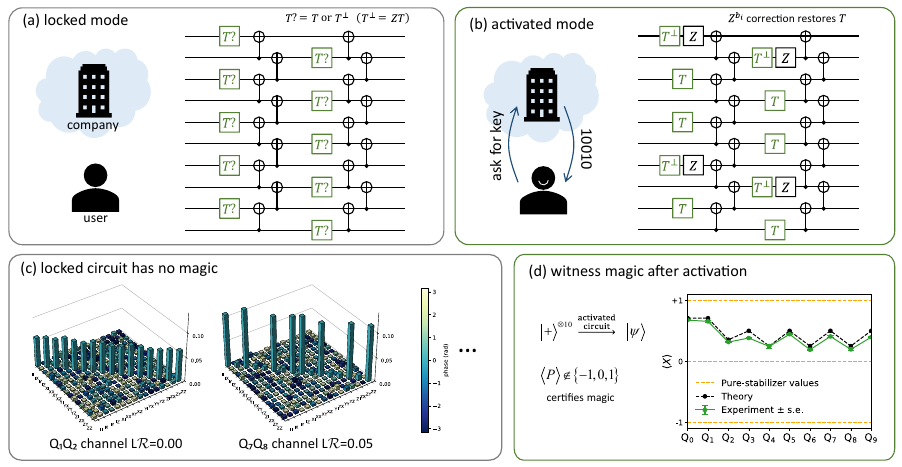}
\caption{
Magic activation key.
A classical secret key switches a Clifford+$T$ processor between a magic-free locked mode and n activated mode \zw{that restores the intended computation}, without quantum communication or shared entanglement.
(a) In the locked mode, each protected $T$ gate use implements either $T$ or $T^\perp\coloneqq ZT$ \zw{according to the key unknown to the user}.
Without the key, hidden randomization (possibly correlated across gates) makes user's effective protected dynamics a stabilizer channel.
(b) Once the key bit $b_i$ is revealed, the user applies the Pauli correction $Z^{b_i}$ after the corresponding protected gate, restoring the intended $T$ operation.
The figure shows a two-gate mask: one key bit makes the same hidden choice, $T$ or $T^\perp$, at two protected  locations, and applying the corresponding $Z^{b_i}$ correction after each location restores both $T$ gates.
(c) Experimental characterization of the displayed 10-qubit brickwall circuit. 
Selected adjacent two-qubit induced channels in the locked mode have near-zero Choi-state log-robustness of magic, $\lrom(\cdot)\coloneqq\mathrm{ln}(\mc{R}(\cdot))$, consistent with the predicted magic-free dynamics on the probed subsystems. See Methods for the faithfulness and scope of this diagnostic.
Complete process-tomography results for all nine adjacent pairs are reported in Appendix~\ref{app:activation key-tomography}.
(d) \zw{After activation, the measured $\langle X\rangle$ values for the output obtained from $\ket{+}^{\otimes 10}$ agree with the theoretical predictions for the restored circuit, providing an operational signature of the recovered $T$ operations.}
}
\label{fig:activationkey}
\end{figure*}

The proof is given in Appendix~\ref{app:proof_of_two_way}.
For a pure single-qubit target, this neutralization mechanism has a simple Bloch-sphere picture: mixing the two-copy target branch with a pair of stabilizer states opposite to the closest stabilizer direction places each local marginal exactly on the boundary of the stabilizer octahedron, and measuring along the same direction then filters out the stabilizer branch and reveals the target state on the other side.
The following theorem further shows that this construction is optimal (complete proof in Appendix~\ref{app:two_way_optimal_success_prob}):

\begin{thm}[Optimal two-qubit separable ENM state]\label{thm:two_way_optimal}
\zw{Let $\ket{\psi}$ be a pure state satisfying
$\ketbra{\psi}{\psi}\notin\stab_1$ with Bloch vector $\mathbf t$, }
and let $\ket{\phi}$ be a stabilizer state maximizing $|\bracket{\phi}{\psi}|^2$, with orthogonal complement $\ket{\phi^\perp}$.
Then
\begin{equation}
\rho_\psi^\leftrightarrow \coloneqq \frac{\|\mathbf t\|_1-1}{\|\mathbf t\|_1+1} \ketbra{\phi^\perp}{\phi^\perp}^{\otimes 2} + \frac{2}{\|\mathbf t\|_1+1} \ketbra{\psi}{\psi}^{\otimes 2}
\end{equation}
is a two-way separable ENM state and attains the optimal one-round local stabilizer extraction probability $p_{\mathrm{success}}=\frac{1+\|\mathbf t\|_\infty}{1+\|\mathbf t\|_1}$ among all permutation-symmetric two-qubit separable ENM states.
\end{thm}

\begin{exmp}
For $\ket{\psi}=\ket{T}$, Theorem~\ref{thm:two_way_optimal} gives
\begin{equation}
\rho_T^{\leftrightarrow}\coloneqq(3-2\sqrt{2})\ketbra{-}{-}^{\otimes 2}+\,\bigl(1-(3-2\sqrt{2})\bigr)\ketbra{T}{T}^{\otimes 2}.
\end{equation}
Measuring either qubit in the $X$ basis and postselecting on the $+$ outcome prepares $\ket{T}$ on the other qubit with probability $1/\sqrt2$.
\end{exmp}

Hence, producing $N$ copies of $\ket{T}$ consumes, in expectation, $\sqrt2\,N$ copies of $\rho_T^{\leftrightarrow}$.
The analogous two-way $F$-type state extracts $\ket{F}$ with probability $1/\sqrt3$; see Example~\ref{exmp:two-way-F-state} in Appendix~\ref{app:two_way_optimal_success_prob}.

A particularly notable two-qubit example combines bidirectional magic extraction with extremal global magic: the ``golden" state $\rho_{\mathrm{golden}}$ (see Methods) is separable and ENM, yet attains a robustness of magic (RoM) of
$\mc{R}(\rho_{\mathrm{golden}})=\sqrt{5}$, the maximum RoM attainable by any two-qubit state~\cite{Howard2017application}.

\zw{These states provide complementary experimental benchmarks for bidirectional recovery and extremal global magic.
We experimentally prepare the two-way $T$- and $F$-type ENM states and $\rho_{\mathrm{golden}}$ and demonstrate magic extraction to either party on a superconducting quantum processor.
We also prepare the golden state and certify that its global log-RoM is close to the two-qubit maximum $\ln\sqrt{5}$, while both marginals remain almost stabilizer; see Appendix~\ref{app:state-experiments}.}

\subsection{Separable ENM channels}

\zw{Magic can be hidden not only in a state but also in the action of a device.
To promote the label-and-correction mechanism of the one-way state construction to control over quantum operations, we extend it to channels: any magical (formally, non-CSP) target channel $\mc{E}$ can be reversibly embedded, via stabilizer superchannels, into an LOSR---and hence separable---ENM channel $\mc{F}_{\mc{E}}$ with magic-free marginal dynamics for both parties; see Definition~\ref{def:ENM_channel} and Theorem~\ref{thm:separableENMchannel} in Methods.}

Operationally, one applies a random Clifford correction after $\mc{E}$ and stores its label in a classical register.
Forgetting the label averages away the locally accessible magic, whereas reading it and undoing the correction recovers $\mc{E}$ exactly.
This reversible masking mechanism provides the operating principle for the activation key protocol developed next.

\begin{figure}[t]
\centering
\includegraphics[width=0.48\textwidth]{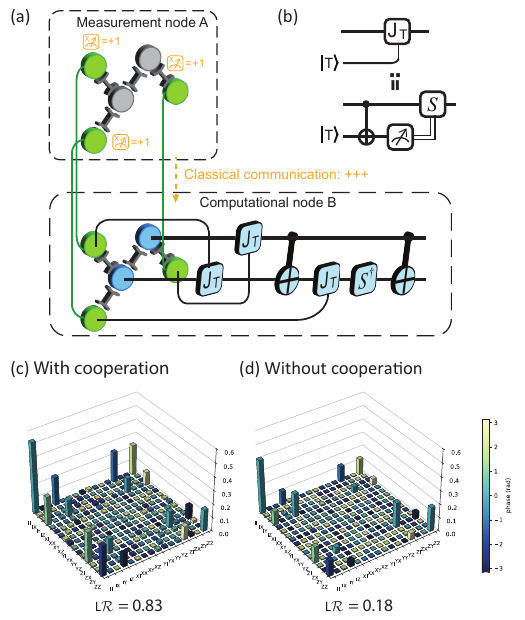}
\caption{
Magic secret sharing with two-way $T$-type ENM links.
(a) Three separable ENM links (green) connect measurement node A to computational node B.
Measuring the link qubits at A in the $X$ basis and postselecting on $X=+1$ for all three, with the outcomes communicated to B, localizes three $\ket{T}$ states at B.
These states are consumed to implement $\mathrm{C}S$ on two target qubits at B.
(b) The standard $T$-state injection gadget $J_T$.
(c,d) Experimental process tomography yields reconstructed Choi-state log-RoM values of $0.83$ with cooperation and $0.18$ without cooperation, which are close to ideal values.
}
\label{fig:two_way_experiment}
\end{figure}

\section{Magic activation key for quantum computational power}\label{sec:activation key}

\zw{The separation between global magic and local access realized by separable ENM channels enables practical applications in quantum resource control.
As a first example, we develop a classical activation key protocol for protected non-Clifford operations.
For each protected use, a random classical label serves as the key: withholding it leaves the user with key-averaged magic-free dynamics, whereas revealing it specifies the stabilizer correction needed to recover the intended operation exactly.
We instantiate this protocol for Clifford+$T$ circuits, requiring neither quantum communication nor shared entanglement.}

Fault-tolerant quantum architectures  combine a high-fidelity Clifford backbone with magic resources supplied by e.g.~distillation or cultivation~\cite{Bravyi2005universal,Bravyi2012distillation,Campbell2017Roads,gidney2024magicstatecultivationgrowing}.
Because magicless computation is efficiently classically simulable, controlling access to magic provides a natural way to meter a processor's quantum computational capability.

\zw{For each protected $T$ gate, the device uses a uniformly random bit $b_i$ and implements $Z^{b_i}T$ while keeping $b_i$ hidden from the user.
Writing $[U]$ for the channel $[U](\rho)=U\rho U^\dagger$ and defining $T^\perp\coloneqq ZT$, averaging over the hidden bit gives the effective channel available without the key:
\begin{equation}
\frac{1}{2}[T]+\frac{1}{2}[T^\perp]
=
\frac{1}{2}[\mbb{I}]+\frac{1}{2}[Z],
\label{eq:locked-T-main}
\end{equation}
which is complete dephasing in the computational basis and hence a stabilizer channel.
Once $b_i$ is released, the user applies $Z^{b_i}$ and recovers the intended $T$ gate.
Withholding or releasing the key therefore switches the user's effective processor between a magic-free mode and the intended Clifford+$T$ computation.
See Methods for the complete operational setting, including key management and an access-metering model analogous to software licensing.}

{Formally, \zw{when considered jointly with the classical key register}, each masked gate is an LOSR and hence separable ENM channel: the classical key does not create magic but helps unlock magic encoded in the channel's global correlations.
Activation uses only classical communication and stabilizer corrections, requiring neither quantum communication nor shared entanglement.
Moreover, one classical bit can mask a pair of $T$ gates within an otherwise stabilizer circuit, even when they occur in different circuit layers or are causally connected (see Methods and Appendix~\ref{app:ActivationKey}).
The same activation key protocol extends to non-Clifford controlled-phase rotations used in the quantum Fourier transform underlying Shor's algorithm~\cite{Shor1994Algorithms}, as well as Pauli rotations that form standard building blocks of variational NISQ circuits such as QAOA~\cite{Cerezo2021variational,Bharti2022NISQ}.
See Methods for the general construction and Appendix~\ref{app:activation key-gate-hiding} for explicit derivations.}

\zw{We experimentally implement both operating modes on a superconducting quantum processor
(Fig.~\ref{fig:activationkey}(c, d)).
In the locked mode, process tomography of the key-averaged induced channels on all nine adjacent qubit pairs gives reconstructed Choi log-RoM values near zero.
This provides a local experimental diagnostic, while CSP of the ideal full ten-qubit channel follows analytically.
In the activated mode, the revealed key restores the intended $T$ gates, and the measured single-qubit $\langle X\rangle$ values for the input $\ket{+}^{\otimes 10}$ agree well with the theory (see Methods).}

\section{Magic secret sharing through cooperative localization}

\zw{Distributed protocols commonly rely on entanglement to make information or operations accessible collectively but not to any party alone, as in distributed computation, data hiding, and secret sharing~\cite{
grover1997quantumtelecomputation,Cirac1999Distributed,
Main2025Distributed,Terhal2001Hiding,DiVincenzo2002data,
Hillery1999secret,Gottesman2000Theory,Zhang2005Multiparty}.
This has motivated a sustained search for quantum phenomena and operational advantages that do not require entanglement~\cite{
Bennett1999nonlocality,GUO2003QSSnoENT,
Braun2018Quantumenhanced,Halder2019Strong}.
Our recoverable separable ENM states bring this possibility to computational power: each node is locally magic-free, yet a stabilizer measurement at one node, followed by communication of its outcome, can localize the magic encoded in their joint correlations at another node.
Classical communication therefore unlocks, rather than supplies, the encoded magic.
We call this cooperative resource-access mechanism \emph{magic secret sharing}.}

\zw{In this resource-access sense, an ENM link plays a role analogous to that of an EPR link.
An EPR pair supports teleportation and nonlocal gates, whereas an ENM link has a more specialized function: it controls where a prescribed nonstabilizer resource becomes available without giving either endpoint locally accessible magic.
Accordingly, $\rho_T^{\rightarrow}$ acts as a deterministic directed access link, whereas $\rho_T^{\leftrightarrow}$ provides probabilistic access in either direction.}

We demonstrate this mechanism by compiling the non-Clifford gate
$\mathrm{C}S=\diag(1,1,1,i)$ using three two-way $T$-type ENM links
(Fig.~\ref{fig:two_way_experiment}).
Measurements in the $X$ basis at one end, followed by communication of the successful outcomes, localize three $\ket{T}$ states at the other end for magic state injection.
If the measurements are not performed, or if their outcomes are ignored, the receiving block retains its stabilizer marginal and the induced process remains a stabilizer channel.
The non-Clifford operation therefore appears only in a classically selected branch of an otherwise magic-free local process.

\zw{We implement this protocol on a superconducting quantum processor using two spatially separated five-qubit blocks connected by three two-way $T$-type ENM links.
Conditioned on the successful upper-block outcomes, the three extracted $\ket{T}$ states supply the magic state injections needed to implement $\mathrm{C}S$ on the two remaining lower-block qubits.
Process tomography shows a clear contrast: the reconstructed Choi state log-RoM is high for the cooperative branch but remains close to zero for the no-cooperation control, demonstrating that cooperation makes the encoded magic locally accessible (see Fig.~\ref{fig:two_way_experiment} and Methods for experimental details).
Together, the two experiments realize complementary forms of ENM-based resource control: a classical key determines whether magic is accessible, while cooperation determines where it is localized.}

\section{Discussion}

\zw{Magic supplies quantum computational power, yet our work shows that entanglement is not required to distribute access to it: the resource can reside entirely in separable correlations, globally present but absent from each local party.
Our reversible ENM embeddings show that this separation is universal across magic-bearing states and channels.
By harnessing this mechanism, we develop magic activation key and  secret sharing protocols that place access to nonclassical computational power under classical or cooperative control.
Furthermore, we experimentally realize representative ENM states and ENM-based resource-control protocols on a superconducting quantum processor.
ENM thereby provides a new foundation for controlling access to quantum computational power in distributed settings.
}

These findings open a broader program of computational power manipulation in quantum networks.
The multipartite constructions in Methods show that extensive magic can be hidden from every party's local view with small cost, motivating broader families of many-body ENM states and channels with scalable preparation and recovery.
\zw{A natural complementary direction is to allow each local view to contain a classically simulable amount of magic and investigate how our results extend to this relaxed setting.
For practical applications, it will be important to characterize the trade-offs among key length, classical communication cost, recovery probability, and robustness to noise in ENM embedding and recovery protocols.
Taken together, our results reveal a broader design principle for distributed quantum computation: separable correlations can hide computational power, while classical information and cooperation determine when, where, and by whom it can be accessed.}


\section*{Methods}

\subsection*{Key preliminaries and definitions}

The $n$-qubit Clifford group $\mathcal{C}_n$ is the subgroup of the $n$-qubit unitary group generated by Hadamard, phase, and CNOT gates. States generated by applying an \(n\)-qubit Clifford unitary to \(\ket{0}^{\otimes n}\) are called pure stabilizer states. The convex hull of the \(n\)-qubit pure stabilizer states, denoted by \(\stab_n\), is taken as the free set in the resource theory of magic. We further write \(\stab=\bigcup_{n\in\mathbb{Z}_{\geq 1}}\stab_n\). Computational protocols composed of stabilizer state preparations, Clifford unitaries, Pauli measurements, and classical feedforward can be efficiently simulated on a classical computer in time polynomial in the number of qubits and the protocol size.

\begin{defn}\label{def:ENM_state}
Consider systems $A \cong (\mathbb{C}^2)^{\otimes n}$ and $B \cong (\mathbb{C}^2)^{\otimes m}$.
A bipartite state $\rho$ acting on $A\otimes B$ is called \emph{locally magic-free} across $A:B$ if its marginals satisfy $\rho_A\in\stab_n$ and $\rho_B\in\stab_m$.
It is called an \emph{entirely nonlocal magic} (ENM) state\footnote{For $\ket{\psi}$ on $A\otimes B$ and a magic measure $\mc{M}$, the term ``non-local magic" has recently also been used to refer to $\min_{U_A,U_B}\mc{M}(U_A\otimes U_B\ket{\psi})$~\cite{cao2025gravitationalbackreactionmagical,Qian2025nonlocal,robin2025antiflatnessnonlocalmagictwoparticle,ahmad2025experimentaldemonstrationnonlocalmagic}, which is a different notion from ENM.
We use ``entirely" to distinguish the two notions and to emphasize that every local reduced state is magic-free.} if it is locally magic-free and globally nonstabilizer, i.e., $\rho\notin\stab_{n+m}$.
\end{defn}

A bipartite state $\rho$ is called \emph{separable} if it can be written as a convex combination of product states,
\begin{equation}
\rho = \sum\nolimits_{i} p_i\, \omega_i^{A} \otimes \omega_i^{B},
\end{equation}
where $\{p_i\}$ is a classical probability distribution.
A \emph{separable ENM state} is a state that is both separable and ENM.

\subsection*{Exponential robustness of magic in separable ENM states}

The reversible one-way construction in Theorem~\ref{thm:separableENM} can hide an exponential amount of magic without entanglement.
\begin{cor}
There exists a family $\{\rho_n\}$ of bipartite separable ENM states on $2n$ qubits, with $n$ qubits per party, whose RoM satisfies
\begin{equation}
\mc{R}(\rho_n)\geq\left(\frac{4}{\ln 2}-o(1)\right)\frac{2^n}{n^2}=2^{n-2\log_2 n+O(1)}.
\end{equation}
\end{cor}

\begin{proof}
The proof of Theorem~4 in Ref.~\cite{Liu2022manybody} gives an $n$-qubit pure state $\sigma_n$ satisfying
\begin{equation}
\mD_{\min}(\sigma_n)\geq n-2\log_2 n+1-\log_2(\ln 2)-o(1),
\end{equation}
where $\mD_{\min}=-\log_2 F_{\stab}$ for pure states.
This gives $\mc{R}(\sigma)\geq2^{\mD_{\min}(\sigma)+1}-1$.
Applying Theorem~\ref{thm:separableENM} to the uniform ensemble of Pauli-$Z$ strings yields a $2n$-qubit separable ENM state $\rho_n$ and deterministic stabilizer protocols satisfying $\mc{E}_1(\sigma_n)=\rho_n$ and $\mc{E}_2(\rho_n)=\sigma_n$.
Thus $\mc{R}(\rho_n)=\mc{R}(\sigma_n)$, which proves the claim.
\end{proof}

For comparison, Ref.~\cite{Liu2022manybody} bounds the maximum RoM over all $2n$-qubit states between $\Omega(2^{2n}/n^2)$ and $O(2^{2n})$.
We conjecture that every bipartite separable ENM state with $n$ qubits per party has RoM at most $O(2^n)$.

\subsection*{Golden ENM state}

We call the following state ``golden'' because the minimum amount of stabilizer noise required to erase its magic is $(\sqrt{5}-1)/2\approx0.618$, the inverse of the golden ratio.
For a single-qubit Bloch vector $(x,y,z)$ with $x^2+y^2+z^2\leq 1$, write
\begin{equation}
\sigma(x,y,z)\coloneqq\tfrac{1}{2}\left(\mbb{I}+xX+yY+zZ\right).
\end{equation}

\begin{exmp}[Golden ENM state]\label{exmp:golden-enm}
The two-qubit state
\begin{equation}
\begin{aligned}
\rho_{\mathrm{golden}}\coloneqq{}&
\tfrac{1}{2}
\sigma\!\bigl(\tfrac{1}{\sqrt{5}},\tfrac{1}{\sqrt{5}},
\tfrac{\sqrt{3}}{\sqrt{5}}\bigr)
\otimes
\sigma\!\bigl(\tfrac{1}{\sqrt{3}},\tfrac{1}{\sqrt{3}},
\tfrac{1}{\sqrt{3}}\bigr)\\
&+\tfrac{1}{2}
\sigma\!\bigl(\tfrac{1}{\sqrt{5}},\tfrac{1}{\sqrt{5}},
-\tfrac{\sqrt{3}}{\sqrt{5}}\bigr)
\otimes
\sigma\!\bigl(-\tfrac{1}{\sqrt{3}},-\tfrac{1}{\sqrt{3}},
-\tfrac{1}{\sqrt{3}}\bigr)
\end{aligned}
\label{eq:golden-state}
\end{equation}
is separable and ENM, and attains the maximum possible two-qubit robustness of magic $\mc{R}(\rho_{\mathrm{golden}})=\sqrt{5}$.
\end{exmp}

The two states appearing on subsystem $B$ in this decomposition are antipodal face states, and hence are orthogonal pure states with maximal one-qubit RoM.

The local marginals of $\rho_{\mathrm{golden}}$ are
\begin{equation}
(\rho_{\mathrm{golden}})_A
=\sigma\!\bigl(\tfrac{1}{\sqrt{5}},\tfrac{1}{\sqrt{5}},0\bigr),
\qquad
(\rho_{\mathrm{golden}})_B
=\tfrac{\mbb{I}}{2}.
\label{eq:golden-marginals}
\end{equation}
Both marginals lie in $\stab_1$.

Moreover, measuring either qubit in the $Z$ basis and communicating the outcome localizes the state's magic onto the other qubit, yielding a distillable noisy face magic state.
The exact RoM certificate and full operational details are given in Appendix~\ref{app:golden-details}.

\subsection*{ENM channels without entanglement}

To extend ENM from states to dynamics, we regard completely stabilizer-preserving (CSP) maps as magic-free channels and separable maps as the entanglement-free class.

A CPTP map $\mc{E}:L((\mbb{C}^2)^{\otimes n})\rightarrow L((\mbb{C}^2)^{\otimes m})$ is called completely stabilizer-preserving (CSP)~\cite{Seddon2019quantifying,Seddon2021Quantifying,Heimendahl2022axiomatic} if, for every ancillary $k$-qubit system,
\begin{equation}
(\mc{E}\otimes\mc{I}_k)(\stab_{n+k})\subset\stab_{m+k},
\end{equation}
where $\mc{I}_k$ is the identity channel on the ancilla.
We use $\mathrm{CSP}$ to denote the set of all such maps.

\begin{defn}[ENM channel]\label{def:ENM_channel}
Consider systems $A \cong (\mathbb{C}^2)^{\otimes n}$ and $B \cong (\mathbb{C}^2)^{\otimes m}$.
A bipartite completely positive trace-preserving (CPTP) map (i.e., a quantum channel) $\mc{E}$ acting on $A\otimes B$ is called \emph{locally magic-free} across $A:B$ if its reduced CPTP maps
\begin{align}
\mc{E}_A:&\rho_A\mapsto\Tr_B\!\left[\mc{E}\!\left(\rho_A\otimes\tfrac{\mbb{I}_{B}}{2^m}\right)\right],\\
\mc{E}_B:&\rho_B\mapsto\Tr_A\!\left[\mc{E}\!\left(\tfrac{\mbb{I}_{A}}{2^n}\otimes\rho_B\right)\right]
\end{align}
belong to $\mathrm{CSP}$.
It is called an \emph{ENM channel} if it is locally magic-free but globally non-CSP, i.e., $\mc{E}\notin\mathrm{CSP}$.
\end{defn}

A bipartite CPTP map $\mc{E}$ on $A\otimes B$ is called \emph{separable}~\cite{watrous2018theory} if there exist CP maps $\{\Phi_j\}_{j=1}^J$ on $A$ and $\{\Psi_j\}_{j=1}^J$ on $B$ such that
\begin{equation}
\mc{E}=\sum\nolimits_{j=1}^J\Phi_j\otimes\Psi_j.
\end{equation}
Such a channel maps every separable input state to a separable output state.
The construction below in fact satisfies the stronger LOSR (local operations and shared randomness) condition: it is a convex combination of product channels coordinated by shared classical randomness.

Channel-state duality gives a useful equivalent characterization.
Suppose $\mc{E}$ is a CPTP map acting on $A\otimes B$, and let $\Lambda_{\mc{E}}\coloneqq(\mc{I}\otimes\mc{E})(\ketbra{\Phi^+}{\Phi^+})$ denote its normalized Choi state~\cite{watrous2018theory}, where $\ket{\Phi^+}=d^{-1/2}\sum_{x=0}^{d-1}\ket{x}\ket{x}$ and $d=\dim(A\otimes B)$.
Let $A'$ and $B'$ denote the reference systems corresponding to $A$ and $B$, respectively.

\begin{fact}\label{fact:enm-channel-choi}
$\mc{E}$ is a separable ENM channel iff $\Lambda_{\mc{E}}$ is a separable ENM state with respect to the partition $A'A:B'B$.
\end{fact}
The proof is given in Appendix~\ref{app:enm-channel-choi}.

The following result is the channel analogue of Theorem~\ref{thm:separableENM}.
A superchannel maps quantum channels to quantum channels~\cite{Chiribella2008Transforming,Gour2019Comparison}, and we call it a stabilizer superchannel when it can be implemented using stabilizer protocols.

\begin{thm}[One-way separable ENM channel]\label{thm:separableENMchannel}
Let $\mc{E}$ be a non-CSP $n$-qubit channel.
Suppose a Clifford ensemble $\{p_i,W_i\}_{i=0}^{M-1}$ satisfies $\sum_i p_i[W_i]\circ\mc{E}\in\mathrm{CSP}$.
Set $m=\lceil\log_2M\rceil$ and $X(i)\coloneqq\bigotimes_{j=1}^mX^{i_j}$, where $(i_1,\cdots,i_m)$ is the $m$-bit expansion of $i$.
Then the $(m+n)$-qubit channel
\begin{equation}\label{eq:defF}
\mc{F}\coloneqq \sum_i p_i\,[X(i)]\otimes\big([W_i]\circ\mc{E}\big)
\end{equation}
is LOSR and ENM.
Moreover, there exist stabilizer superchannels $\Theta_1,\Theta_2$ such that $\Theta_1(\mc{E})=\mc{F}$ and $\Theta_2(\mc{F})=\mc{E}$.
\end{thm}

The label marginal of $\mc{F}$ is $\sum_i p_i[X(i)]\in\mathrm{CSP}$, whereas its system marginal is the CSP average assumed in the theorem.
Initializing the label register in $\ket{0}^{\otimes m}$, reading its computational-basis output and applying $W_i^\dagger$ recovers $\mc{E}$ exactly; the complete proof is given in Appendix~\ref{app:proof-separable-enm-channel}.

Let $\mathcal{P}_n^+\coloneqq\{I,X,Y,Z\}^{\otimes n}$ denote the set of
$n$-qubit Pauli operators with phase $+1$.
Uniformly averaging all output Pauli corrections gives
\begin{equation}\label{eq:universal-channel-pauli-twirl}
\frac{1}{4^n}\sum_{P\in\mathcal{P}_n^+}[P]\circ\mc{E}=\mc{D},\quad\mc{D}(\rho)=\Tr(\rho)\frac{\mbb{I}}{2^n}.
\end{equation}
Therefore, every non-CSP $n$-qubit channel admits a reversible LOSR ENM embedding on $3n$ qubits in total.
More economical Clifford ensembles, when available, reduce the size of the classical label register.

\subsection*{Activation key protocol and extensions}\label{sec:methods-activation key}

We detail the activation key protocol introduced in Sec.~\ref{sec:activation key}, including its channel formulation, key compression by correlated randomization, and extensions beyond $T$ gates.

\paragraph{Operational setting and key management.}
Consider a logical Clifford+$T$ processor whose protected $T$ gates are masked using classical key information held by the provider.
This mechanism parallels software licensing: by releasing key bits on demand, the provider can grant metered access to the device's magic operations, for example on a pay-per-use or subscription basis.

Suppose the processor has a lifetime capacity of at most $N$ protected logical $T$-gate operations.
In the elementary scheme, the provider samples a uniformly random bit string
\begin{equation}
\mathbf{b}=(b_1,\cdots,b_N)\in\{0,1\}^N,\quad\mathbf{b}\sim\mathrm{Unif}(\{0,1\}^N)
\end{equation}
and loads it into trusted storage accessible to the device but not the user.
The key remains available to the device but hidden from the user, for instance through a hardware security module.
At the $i$-th protected $T$ gate, the hardware applies $Z^{b_i}$ immediately after $T$.
If the provider releases $b_i$, the user applies the same Pauli correction and recovers the intended gate.

\paragraph{ENM channel interpretation of the activation key protocol.}
Let $K$ be a one-qubit key register and $S$ the system register, and let $T^\perp\coloneqq ZT$.
Define
\begin{equation}
\mc{F}_T\coloneqq\frac{1}{2}[\mbb{I}_K\otimes T_S]+\frac{1}{2}[X_K\otimes T^\perp_S].\label{eq:activation-ENM-channel}
\end{equation}
$\mc{F}_T$ is LOSR and hence separable.
Its reduced channels on the key and system registers are
\begin{align}
(\mc{F}_T)_K&=\frac{1}{2}[\mbb{I}]+\frac{1}{2}[X]\in\mathrm{CSP},\label{eq:activation key-reduced-K}\\
(\mc{F}_T)_S&=\frac{1}{2}[\mbb{I}]+\frac{1}{2}[Z]\in\mathrm{CSP}.\label{eq:activation key-reduced-S}
\end{align}
The stabilizer decoding protocol in Theorem~\ref{thm:separableENMchannel} recovers the $T$ gate from $\mc{F}_T$, so $\mc{F}_T\notin\mathrm{CSP}$.
Therefore, $\mc{F}_T$ is a separable ENM channel.

Operationally, initialize the key register in $\ket{0}$.
For every input state $\rho$ on $S$,
\begin{align}
&\mc{F}_T\left(\ket{0}\!\bra{0}_K\otimes\rho_S\right)\\
=&\frac{1}{2}\ket{0}\!\bra{0}_K\otimes T\rho T^\dagger+\frac{1}{2}\ket{1}\!\bra{1}_K\otimes T^\perp\rho(T^\perp)^\dagger.
\end{align}

Thus, in this operational realization, $K$ carries a uniformly random classical bit $b$, while the system undergoes $Z^bT$.
If $b$ is withheld, the user's effective channel is the complete dephasing channel in Eq.~\eqref{eq:activation key-reduced-S}. If $b$ is revealed, the user can apply $Z^b$ and recover the intended $T$ gate.
Because $K$ is initialized in $\ket{0}$ and remains diagonal in the computational basis, the key register and its processing can be implemented entirely using classical storage, control, and communication.
This classical procedure implements the stabilizer decoder of Theorem~\ref{thm:separableENMchannel}; see Appendix~\ref{app:activation key-single-gate-decoder}.

\paragraph{Key compression in states and channels.}

A single key bit can lock/unlock more than one magic gate.
This compression already appears at the state level: with $\ket{T^\perp}=Z\ket{T}$, the state
\begin{equation}
\rho_{TT}^{\rightarrow}\coloneqq\frac{1}{2}\ketbra{0}{0}\otimes\ketbra{TT}{TT}+\frac{1}{2}\ketbra{1}{1}\otimes\ketbra{T^\perp T^\perp}{T^\perp T^\perp}
\end{equation}
is a one-way separable ENM state.
Its key marginal is $\mbb{I}_2/2$, and its data marginal $\frac{1}{2}\ketbra{TT}{TT}+\frac{1}{2}\ketbra{T^\perp T^\perp}{T^\perp T^\perp}$ is a stabilizer state.

Measuring the key and applying $Z^b$ to both data qubits recovers $\ket{TT}$ deterministically, so one classical bit locks/unlocks two $T$ states.

The same compression extends to channels.
Although the elementary construction assigns one bit to each protected $T$-gate location, two $T$ gates can share a single bit through the correlated randomization
\begin{equation}
\mc{F}_{T,T}\coloneqq\frac{1}{2}[\mbb{I}\otimes T^{\otimes 2}]+\frac{1}{2}[X\otimes(T^\perp)^{\otimes 2}].\label{eq:activation-two-T-channel}
\end{equation}
The reduced channel on the two system qubits is
\begin{equation}\label{eq:activation-two-T-reduced}
\frac{1}{2}[T^{\otimes 2}]+\frac{1}{2}[(T^\perp)^{\otimes 2}],
\end{equation}
which is CSP (see Appendix~\ref{app:activation key-gate-hiding}).

Operationally, when the common key bit is hidden, the two protected $T$ gates are jointly replaced by the CSP channel in Eq.~\eqref{eq:activation-two-T-reduced}.
When the bit is released, applying the same Pauli correction $Z^b$ after each protected gate restores both gates exactly.

By Proposition~\ref{prop:Choi_stab_thus_no_capacity} in Appendix~\ref{app:ActivationKey}, the locked circuit remains CSP when the two locked $T$ gates are located in an arbitrary stabilizer circuit, even if they occur in different layers or are causally connected.

\paragraph{Extensions beyond $T$ gates.}

The activation key construction extends beyond $T$ gates.
In general, if $U$ is an $n$-qubit unitary diagonal in the computational basis, then a complete $Z$ twirl gives
\begin{equation}
\frac{1}{2^n}\sum_{\mathbf b\in\{0,1\}^n}[Z(\mathbf b)]\circ[U]=\frac{1}{2^n}\sum_{\mathbf b\in\{0,1\}^n}[Z(\mathbf b)],\label{eq:diagonal-gate-activation}
\end{equation}
where $Z(\mathbf b)\coloneqq\bigotimes_{j=1}^n Z^{b_j}$.
The right-hand side is complete computational basis dephasing.
Thus, $n$ classical bits are sufficient to hide any diagonal unitary $U$ whose channel $[U]$ is non-CSP in a separable ENM channel.

Important gates often require fewer bits.
For example:
\begin{align}
\frac{1}{2}[\mathrm{CC}Z]+\frac{1}{2}[(Z\otimes\mbb{I}\otimes\mbb{I})\mathrm{CC}Z]&\in\mathrm{CSP},\\
\frac{1}{2}[\mathrm{C}S]+\frac{1}{2}[(Z\otimes\mbb{I})\mathrm{C}S]&\in\mathrm{CSP},\\
\frac{1}{2}[\mathrm{C}R(\theta)]+\frac{1}{2}[\mathrm{C}Z\cdot\mathrm{C}R(\theta)]&\in\mathrm{CSP},\\
\left(\frac{1}{2}[\mbb{I}]+\frac{1}{2}[P]\right)\circ [R_P(\theta)]&\in\mathrm{CSP},\label{eq:Pauli-rotation-activation}
\end{align}
where $\mathrm{C}R(\theta)=\diag(1,1,1,e^{i\pi\theta})$ is the controlled-phase rotation appearing in quantum Fourier transform circuits and hence in Shor's algorithm~\cite{Shor1994Algorithms}, and $R_P(\theta)\coloneqq e^{-i\theta P/2}$ is a rotation generated by an arbitrary Pauli string $P$ and often appears in NISQ algorithms~\cite{Cerezo2021variational,Bharti2022NISQ}.
The detailed derivations are collected in Appendix~\ref{app:activation key-gate-hiding}.

\subsection*{Multipartite separable ENM states}

Relative to a specified partition $A_1:\cdots:A_n$, a multipartite ENM state is globally nonstabilizer while every single-party marginal is stabilizer.
We give two fully separable ENM $n$-qubit families.

Let $\ket{\psi}$ be a pure one-qubit nonstabilizer state with Bloch vector $\mathbf t$, and let $\ket{\phi}$ be a stabilizer state maximizing $|\bracket{\phi}{\psi}|^2$.
For $n\ge2$, define
\begin{equation}\rho_\psi^{(n)}\coloneqq \frac{\|\mathbf t\|_1-1}{\|\mathbf t\|_1+1}\ketbra{\phi^\perp}{\phi^\perp}^{\otimes n}+\frac{2}{\|\mathbf t\|_1+1}\ketbra{\psi}{\psi}^{\otimes n}.\label{eq:multipartite-rho-psi}\end{equation}
It has stabilizer one-qubit marginals, while projecting any party onto $\ket{\phi}$ prepares $\ket{\psi}^{\otimes(n-1)}$ with probability $(1+\|\mathbf t\|_\infty)/(1+\|\mathbf t\|_1)$.

For $n\ge3$, a stronger family, whose one- and two-qubit reductions are all magic-free, is
\begin{equation}
\rho_T^{(n)}
\coloneqq
\frac{1}{2}\ketbra{T}{T}^{\otimes n}
+\frac{1}{2}\ketbra{T^\perp}{T^\perp}^{\otimes n}.
\label{eq:multipartite-rho-T}
\end{equation}
Every one-qubit marginal is maximally mixed, whereas every two-qubit marginal is an equal mixture of the stabilizer states
$(\ket{00}+i\ket{11})/\sqrt{2}$ and
$(\ket{01}+\ket{10})/\sqrt{2}$.
Measuring any two qubits in the $X$ basis yields the $++$ outcome with probability $3/8$.
Conditioned on this outcome, the reduced state of each unmeasured qubit has Bloch vector $(2/3,2/3,0)$, enabling the Bravyi--Kitaev distillation~\cite{Bravyi2005universal}.

For the stabilizer R\'enyi entropy (SRE) $M_\alpha$~\cite{Lorenzo2022stabilizer}, both families obey $M_\alpha(\rho^{(n)})=nM_\alpha(\ketbra{\chi}{\chi})+O(1)$ for $1/2\leq\alpha<1$, where $(\rho^{(n)},\ket{\chi})$ denotes either $(\rho_\psi^{(n)},\ket{\psi})$ or $(\rho_T^{(n)},\ket{T})$.
Since $M_\alpha\leq2\log_2\mc{R}$, this also certifies exponentially growing RoM; details are given in Appendix~\ref{app:multipartite-scaling}.

\subsection*{Activation key experiment}

We implement the brickwall circuit in Fig.~\ref{fig:activationkey} on a superconducting quantum processor.
In the locked mode, a key assignment $b\in\{0,1\}^5$ is sampled uniformly and determines which intended $T$ gates are replaced by $T^\perp$.
Let $U_b$ denote the corresponding ideal key-dependent circuit unitary.
For a selected adjacent two-qubit subsystem $S$, let $E$ denote the remaining eight qubits.
We initialize $E$ in $\ket{+}^{\otimes 8}$ and trace it out after applying the circuit.
The ideal key-averaged induced channel is \begin{equation}
\mc{E}_S(\rho_S)
=
\Tr_E\!\left[
\mathbb{E}_{b\sim\mathrm{Unif}(\{0,1\}^5)}
\left[
U_b
\left(
\rho_S\otimes\ketbra{+}{+}^{\otimes 8}
\right)
U_b^\dagger
\right]
\right].
\end{equation}
For each selected adjacent pair, we collect process-tomography data for all 32 key assignments with equal weight, average the data before reconstruction, and evaluate the log-RoM of the resulting Choi state.
This Choi-state diagnostic is faithful for each induced channel: $\lrom(\Lambda_{\mc{E}_S})=0\Longleftrightarrow\Lambda_{\mc{E}_S}\in\stab\Longleftrightarrow\mc{E}_S\in\mathrm{CSP}$.
The first equivalence follows from the faithfulness of RoM, and the second from the Choi characterization of CSP channels~\cite[Lemma~4.2]{Seddon2019quantifying}.
Across all nine adjacent pairs, the reconstructed induced channels exhibit the expected locked-mode signature, with log-RoM consistently close to zero.
Together, these measurements support the predicted magic-free local dynamics across the device.
This local validation complements the analytic full-channel guarantee: Proposition~\ref{prop:Choi_stab_thus_no_capacity} establishes that the ideal key-averaged ten-qubit channel is CSP.
Figure~\ref{fig:activationkey}(c) shows two representative pairs, while the complete results are reported in Appendix~\ref{app:activation key-tomography}.

In the activated mode, the revealed key restores every intended $T$ gate.
Starting from $\ket{+}^{\otimes 10}$, we measure $\langle X\rangle$ on each qubit.
The measured values, shown in Fig.~\ref{fig:activationkey}(d), agree well with the theoretical predictions.

\subsection*{Magic secret sharing experiment}

We implement the ENM-based magic-secret-sharing protocol in Fig.~\ref{fig:two_way_experiment} on a superconducting quantum processor.
Two spatially separated five-qubit subregions of the chip form an upper block and a lower block.
Three cross-block qubit pairs are prepared as two-way $T$-type ENM links, with one qubit from each pair in each block.
Because each link has stabilizer single-qubit marginals, neither block has local access to magic. The computational resource is encoded in the correlations between the blocks.

To activate the shared resource, we measure the three upper-block link qubits in the $X$ basis and postselect on the $+1$ outcome for each measurement.
After these outcomes are communicated, each corresponding lower-block qubit is known to be in $\ket{T}$.
The three extracted states are then consumed by the standard magic-state-injection circuit in Fig.~\ref{fig:two_way_experiment}(b).
The injection circuit thereby realizes a $\mathrm{C}S$ gate on the two remaining lower-block qubits using only stabilizer operations.

To isolate the role of cooperation, we leave the upper-block link qubits unmeasured and run the same injection circuit in the lower block.
The three lower-block link qubits therefore enter the circuit in their stabilizer marginals rather than as localized $\ket{T}$ states, so the ideal control implements a stabilizer channel.
Quantum process tomography yields a reconstructed Choi-state RoM close to the stabilizer value of $1$, consistent with this prediction (Fig.~\ref{fig:two_way_experiment}(d)).
Together with the activated result, this control supports the central mechanism: the lower block has no locally accessible magic on its own, whereas conditioning on the communicated upper-block outcomes localizes the magic encoded in the cross-block correlations.

\begin{acknowledgments}
We thank Huiping Lin, Zhenhuan Liu, Huikai Xu, Zijian Zhang for valuable discussions.
F.W.\ and Z.-W.L.\ are supported in part by NSFC under Grant No.~12475023, Dushi Program, and a startup funding from YMSC.
F.W.\ is supported by the Shuimu Tsinghua Scholar Program.
R.W., Y.Z., and F.Y.\ are supported by the National Natural Science Foundation of China (Grants No.\,12404558, No.\,12322413, No.\,92476206) and Beijing Natural Science Foundation (Grants No.\,JQ25014).
\end{acknowledgments}


%


\onecolumngrid
\appendix

\section{Magic resource-theoretic definitions}

\subsection{Robustness of magic and dual witnesses}

For an $n$-qubit state $\rho$, its robustness of magic (RoM) is
\begin{equation}
\mc{R}(\rho)\coloneqq\min_{\sigma,\tau\in\stab_n}\left\{2a+1\,\middle|\,\rho=(a+1)\sigma-a\tau,\ a\ge0\right\}.\label{eq:RoM_primal_app}
\end{equation}
RoM is faithful: $\mc{R}(\rho)=1$ if and only if $\rho\in\stab_n$; it is monotone under trace-preserving stabilizer operations; it is invariant under tensoring with stabilizer states, $\mc{R}(\rho\otimes\sigma)=\mc{R}(\rho)$ for $\sigma\in\stab$; it is convex in the usual sense.
Operationally, quasiprobability simulation of Clifford circuits assisted by a nonstabilizer state $\rho$ has sampling overhead scaling as $\mc{R}(\rho)^2$~\cite{Howard2017application}.

The dual form of Eq.~\eqref{eq:RoM_primal_app} is:
\begin{equation}\label{eq:RoM_dual}
\begin{aligned}
\mc{R}(\rho)=\textbf{\text{max}}~ & \Tr(\rho A)~~~~~~~~~~\text{over Hermitian matrices }A,\\
\textbf{s.t.}~&\left|\Tr\left(\phi A\right)\right| \leq 1~~\text{for all }\phi\in\stab_n .
\end{aligned}
\end{equation}
Any feasible observable $A$ gives a lower bound $\mc{R}(\rho)\ge\Tr(\rho A)$; if this matches an independent upper bound, $A$ certifies the exact RoM.

For one-qubit states we use the Bloch parametrization
\begin{equation}
\sigma(x,y,z)=\frac{1}{2}\left(\mbb{I}+xX+yY+zZ\right).
\end{equation}
The one-qubit stabilizer polytope is the octahedron
\begin{equation}
\stab_1=\left\{\sigma(x,y,z):|x|+|y|+|z|\le1\right\},\label{eq:stab1-octahedron}
\end{equation}
and the one-qubit RoM is
\begin{equation}
\mc{R}\bigl(\sigma(x,y,z)\bigr)=\max\{1,|x|+|y|+|z|\}.\label{eq:single-qubit-RoM}
\end{equation}

\subsection{Stabilizer protocols and CSP maps}

By a \emph{deterministic stabilizer protocol} we mean a CPTP map built from stabilizer state preparation, Clifford unitaries, Pauli measurements, classical randomness and feedforward, and discarding subsystems.
Postselection is not included in this definition, unless explicitly stated, so that the resulting map is trace-preserving.

A quantum channel $\mc{E}:L((\mbb{C}^2)^{\otimes n})\to L((\mbb{C}^2)^{\otimes m})$ is completely stabilizer-preserving (CSP) if
\begin{equation}
(\mc{E}\otimes\mc{I}_k)(\stab_{n+k})\subseteq \stab_{m+k}
\end{equation}
for every $k$-qubit ancillary system.
Every deterministic stabilizer protocol is CSP.
For one-qubit maps, the extremal structure of $\mathrm{CSP}_1$ used in the optimality proof is stated in Lemma~\ref{lemma:CPS1}.

For a channel $\mc{E}$, let
\begin{equation}
\Lambda_{\mc{E}}\coloneqq(\mc{I}\otimes\mc{E})(\ketbra{\Phi^+}{\Phi^+}),\quad\ket{\Phi^+}=2^{-n/2}\sum_{x\in\{0,1\}^n}\ket{xx},
\end{equation}
be its normalized Choi state.

A CPTP map is CSP if and only if its Choi state is a stabilizer state~\cite[Lemma~4.2]{Seddon2019quantifying}.

\section{Separable ENM states: constructions, recovery, and optimality}\label{app:separable-ENM-state-details}

\subsection{Proof of Theorem~\ref{thm:separableENM}}
\label{app:proof_of_one_way}

Let $m=\lceil \log_2 M\rceil$ and introduce an ancilla register $R\cong (\mathbb{C}^2)^{\otimes m}$ with computational basis $\{\ket{i}\}_{i=0}^{2^m-1}$.
Define the separable state
\begin{equation}\label{eq:rho-def}
\rho_\sigma^{\rightarrow} \coloneqq \sum_{i=0}^{M-1} p_i\, \ketbra{i}{i}_{R}\otimes W_i \sigma W_i^\dagger.
\end{equation}
On the ancilla register, $(\rho_\sigma^{\rightarrow})_R = \sum_{i=0}^{M-1} p_i \ketbra{i}{i}_R$ is diagonal in the computational basis, thus a stabilizer state; on the $n$-qubit system, $(\rho_\sigma^{\rightarrow})_{\mathrm{sys}} = \sum_{i=0}^{M-1} p_i W_i \sigma W_i^\dagger$, which belongs to $\stab$ by hypothesis.

Define $\mc{E}_1$ as: sample $i$ with probability $p_i$ using classical randomness; prepare the $m$-qubit stabilizer state $\ket{i}_R$; apply the Clifford $W_i$ to the system controlled on the classical value $i$.
Then $\mc{E}_1(\sigma)=\rho_\sigma^{\rightarrow}$.

Define $\mc{E}_2$ as: measure $R$ in the computational basis obtaining outcome $i$; apply the Clifford $W_i^\dagger$ to the system conditioned on $i$; discard $R$.
Then $\mc{E}_2(\rho_\sigma^{\rightarrow})=\sigma$.
Because $\mc{E}_1$ and $\mc{E}_2$ are stabilizer protocols, the construction satisfies $\rho_\sigma^{\rightarrow}\in\stab_{m+n}$ if and only if $\sigma\in\stab_n$.
The hypothesis $\sigma\notin\stab_n$ therefore makes $\rho_\sigma^{\rightarrow}$ globally nonstabilizer; together with the stabilizer marginals above, this proves that it is ENM.

\begin{rem}
For any $n$-qubit state $\sigma$, take $M=2^n$, $p_{\mathbf b}=2^{-n}$, and $W_{\mathbf b}=Z(\mathbf b)\coloneqq Z^{b_1}\otimes\cdots\otimes Z^{b_n}$.
Then
\begin{equation}
\frac{1}{2^n}\sum_{\mathbf b\in\{0,1\}^n}Z(\mathbf b)\sigma Z(\mathbf b)
\end{equation}
is diagonal in the computational basis and hence belongs to $\stab_n$.
Thus Theorem~\ref{thm:separableENM} gives a reversible $(n+n)$-qubit one-way separable ENM embedding for every $\sigma\notin\stab_n$.
\end{rem}

\subsection{One-way $T$ state: comparison with entangled encoding}\label{app:Oneway_T_comparasion}

We consider two bipartite one-way ENM realizations of the single-qubit magic state $\ket{T}$: the first is separable, while the second is maximally entangled,
\begin{align}
\rho_T^{\rightarrow}=&\frac{1}{2}\ketbra{0}{0}\otimes\ketbra{T}{T}+\frac{1}{2}\ketbra{1}{1}\otimes\ketbra{T^\perp}{T^\perp},\\
\ket{\Phi_T}=&\tfrac{1}{\sqrt{2}}(\ket{0T}+\ket{1T^\perp}).
\end{align}
Both states have the reduced density matrix $\mbb{I}_2/2$ on each qubit.
Moreover, they have the same value under any magic monotone $\mc{M}$: $\mc{M}(\rho_T^{\rightarrow})=\mc{M}(\ketbra{\Phi_T}{\Phi_T})$.

These two realizations behave differently under noise.
The state $\rho_T^{\rightarrow}$ is a classical--quantum state: its first qubit only stores a classical bit in the computational basis, indicating whether the second qubit is prepared in $\ket{T}$ or $\ket{T^\perp}$.
Therefore, one may store this classical bit without exposing it to quantum noise, and only the second qubit needs to be kept in quantum memory.
By contrast, both qubits of $\ket{\Phi_T}$ must be stored coherently and are therefore both subject to noise.
To compare these two situations, let $\mc{E}_\lambda$ be the single-qubit depolarizing channel, $\mc{E}_\lambda(\sigma) \coloneqq (1-\lambda)\sigma+\lambda\frac{\mbb{I}_2}{2}$, then we have:

\begin{prop}
For $0<\lambda<1-\frac{1}{\sqrt{2}}$, we have (see Fig.~\ref{fig:noisy})
\begin{equation}
\mc{R}\big(\mc{I}\otimes\mc{E}_\lambda\left(\rho_T^{\rightarrow}\right)\big)>\mc{R}\big(\mc{E}_\lambda\otimes\mc{E}_\lambda(\ketbra{\Phi_T}{\Phi_T})\big).
\end{equation}
\end{prop}

\begin{proof}
Since $\ket{T^\perp}=Z\ket{T}$ and $\mc{E}_\lambda$ commutes with single-qubit unitaries, we have
\begin{align}
(\mc{I}\otimes\mc{E}_\lambda)(\rho_T^{\rightarrow})&=(\mc{I}\otimes\mc{E}_\lambda)\left(\frac{1}{2}\ketbra{0}{0}\otimes\ketbra{T}{T}+\frac{1}{2}\ketbra{1}{1}\otimes\ketbra{T^\perp}{T^\perp}\right)\\
&=\frac{1}{2}\ketbra{0}{0}\otimes\mc{E}_\lambda(\ketbra{T}{T})+\frac{1}{2}\ketbra{1}{1}\otimes\mc{E}_\lambda(\ketbra{T^\perp}{T^\perp})\\
&=\frac{1}{2}\ketbra{0}{0}\otimes\mc{E}_\lambda(\ketbra{T}{T})+\frac{1}{2}\ketbra{1}{1}\otimes Z\mc{E}_\lambda(\ketbra{T}{T})Z\\
&=\mathrm{C}Z\left(\frac{\mbb{I}_2}{2}\otimes \mc{E}_\lambda(\ketbra{T}{T})\right)\mathrm{C}Z.
\end{align}
Hence, by the Clifford invariance of the RoM,
\begin{equation}
\mc{R}\big((\mc{I}\otimes\mc{E}_\lambda)(\rho_T^{\rightarrow})\big)=\mc{R}\big(\mc{E}_\lambda(\ketbra{T}{T})\big).
\end{equation}
Since the RoM of a 1-qubit state can be computed by
\begin{equation}
\mc{R}\left(\frac{\mbb{I}+r_xX+r_yY+r_zZ}{2}\right)=\max\{1,|r_x|+|r_y|+|r_z|\},
\end{equation}
we obtain
\begin{equation}
\mc{R}\big(\mc{E}_\lambda(\ketbra{T}{T})\big)=
\begin{cases}
\sqrt{2}(1-\lambda),&\text{when }0\le\lambda\le1-\frac{1}{\sqrt{2}},\\
1,&\text{when }1-\frac{1}{\sqrt{2}}\le\lambda\le1.
\end{cases}
\end{equation}

Now note that $\ket{\Phi_T} = (\mbb{I}_2\otimes TH)\ket{\Phi^+}$, where $\ket{\Phi^+}=\frac{1}{\sqrt{2}}(\ket{00}+\ket{11})$.
Therefore,
\begin{align}
(\mc{E}_\lambda\otimes\mc{E}_\lambda)(\ketbra{\Phi_T}{\Phi_T})&=(\mbb{I}_2\otimes TH)(\mc{E}_\lambda\otimes\mc{E}_\lambda)(\ketbra{\Phi^+}{\Phi^+})(\mbb{I}_2\otimes TH)^\dagger\\
&=(1-\lambda)^2\ketbra{\Phi_T}{\Phi_T}+(2\lambda-\lambda^2)\frac{\mbb{I}_4}{4}.
\end{align}
Hence
\begin{align}
\mc{R}\big((\mc{E}_\lambda\otimes\mc{E}_\lambda)(\ketbra{\Phi_T}{\Phi_T})\big)
&=
\max\left\{
\Tr\left[
\left(
(1-\lambda)^2\ketbra{\Phi_T}{\Phi_T}
+(2\lambda-\lambda^2)\frac{\mbb{I}_4}{4}
\right)A
\right],
1
\right\}\\
&=
\max\left\{
(1-\lambda)^2\sqrt{2}+\frac{2\lambda-\lambda^2}{2},
1
\right\}\\
&=
\begin{cases}
(\sqrt{2}-\tfrac{1}{2})\lambda^2+(1-2\sqrt{2})\lambda+\sqrt{2},
&\text{when }0\le\lambda\le1-\tfrac{1}{\sqrt{2\sqrt{2}-1}},\\
1,
&\text{when }1-\tfrac{1}{\sqrt{2\sqrt{2}-1}}\le\lambda\le1.
\end{cases}
\end{align}
Here
\begin{equation}
A\coloneqq\frac{1}{2}\big(\mbb{I}\otimes\mbb{I}+X\otimes\mbb{I}-Y\otimes X+Z\otimes X+Y\otimes Y+Z\otimes Y+\mbb{I}\otimes Z-X\otimes Z\big)
\end{equation}
is an optimal solution of the dual problem \eqref{eq:RoM_dual}, found numerically, and satisfies
\begin{equation}
\bra{\Phi_T}A\ket{\Phi_T}=\sqrt{2}.
\end{equation}

The decay profile of the RoM of these two states under noise is shown in Fig.~\ref{fig:noisy}.
\end{proof}

\begin{figure}[t]
\centering
\includegraphics[width=0.43\textwidth]{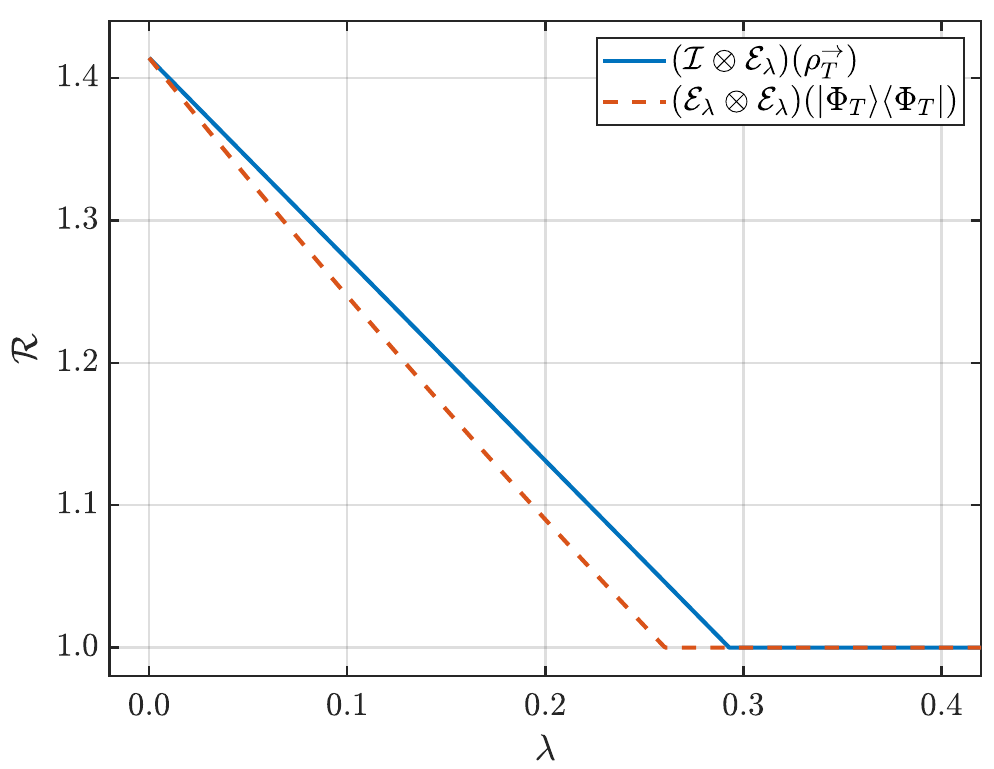}
\caption{Comparison of the RoM of $(\mc{I}\otimes\mc{E}_\lambda)(\rho_T^{\rightarrow})$ and $(\mc{E}_\lambda\otimes\mc{E}_\lambda)(\ketbra{\Phi_T}{\Phi_T})$.
The former is shown by the blue solid line and decays linearly until reaching the threshold $1-\tfrac{1}{\sqrt{2}}$.
The latter is shown by the orange dotted line and decays quadratically until reaching the threshold $1-\tfrac{1}{\sqrt{2\sqrt{2}-1}}$.}
\label{fig:noisy}
\end{figure}

\subsection{Flagged two-way embedding}\label{app:flagged-two-way-embedding}

\begin{thm}[Universal flagged two-way embedding]\label{thm:universal-two-way}
Let $\sigma\notin\stab_n$ and let $\{p_i,W_i\}_{i=0}^{M-1}$ be a Clifford ensemble satisfying $\sum_i p_iW_i\sigma W_i^\dagger\in\stab_n$.
Set $m=\lceil\log_2M\rceil$, $q=\max\{m,n\}$, and $\rho_\sigma^{\rightarrow}\coloneqq\sum_i p_i\ketbra{i}{i}\otimes W_i\sigma W_i^\dagger$.
Let $\widetilde{\rho}_{\sigma,A'B'}^{\rightarrow}$ be obtained by zero-padding its $m$-qubit key and $n$-qubit data registers locally to $q$ qubits.
For flag qubits $a,b$ and the swap $\mathsf{S}_{A'B'}$, define
\begin{equation}
\Omega_{\sigma}^{\leftrightarrow}\coloneqq\frac{1}{2}\left(\ketbra{01}{01}_{ab}\otimes\widetilde{\rho}_{\sigma,A'B'}^{\rightarrow}+\ketbra{10}{10}_{ab}\otimes\mathsf{S}_{A'B'}\widetilde{\rho}_{\sigma,A'B'}^{\rightarrow}\mathsf{S}_{A'B'}\right).
\end{equation}
Then $\Omega_{\sigma}^{\leftrightarrow}$ is permutation symmetric, separable, and ENM on $aA'|bB'$; a deterministic stabilizer protocol maps $\sigma$ to $\Omega_{\sigma}^{\leftrightarrow}$, and either party can recover $\sigma$ with probability $1/2$ using local stabilizer operations and classical communication.
\end{thm}

\begin{proof}
By Theorem~\ref{thm:separableENM}, $\rho_\sigma^{\rightarrow}$ is separable, has stabilizer marginals, and admits deterministic stabilizer encoding and decoding.
Local zero-padding preserves these properties.
Hence both flagged branches of $\Omega_\sigma^{\leftrightarrow}$ are separable, and their equal mixture is permutation symmetric.
Each marginal is a convex mixture of the two padded marginals tagged by computational-basis flag states, so both belong to $\stab_{q+1}$.

To encode $\sigma$, sample the orientation uniformly and run the corresponding padded one-way encoder.
To recover on $bB'$, measure $a$ and accept the outcome $0$, which occurs with probability $1/2$; the padded one-way decoder then recovers $\sigma$ after discarding the padding.
Interchanging the parties gives the reverse protocol.
All steps use local stabilizer operations and one-way classical communication.
\end{proof}

\subsection{Proof of Theorem~\ref{thm:two_way_general}}
\label{app:proof_of_two_way}

The state $\rho^{\leftrightarrow}_{\sigma}$ is separable, and both of its marginals equal $p\sigma+(1-p)\tau\in\stab_n$.

To extract $\sigma$ on subsystem $B$, perform the stabilizer measurement effect $\Pi$ on subsystem $A$.
The unnormalized post-measurement state on subsystem $B$ is
\begin{equation}
p\,\Tr(\Pi\sigma)\,\sigma+(1-p)\,\Tr(\Pi\tau)\,\tau=p\,\Tr(\Pi\sigma)\,\sigma,
\end{equation}
where we used $\Tr(\Pi\tau)=0$.
Thus, conditioned on the successful outcome, the remaining subsystem is exactly $\sigma$, and the success probability is $p\Tr(\Pi\sigma)$.
If $\rho^{\leftrightarrow}_{\sigma}$ were globally stabilizer, this postselected stabilizer measurement would leave a stabilizer state on the unmeasured subsystem, contradicting $\sigma\notin\stab_n$.
Hence $\rho^{\leftrightarrow}_{\sigma}$ is globally nonstabilizer and therefore ENM.
By symmetry, the same protocol extracts $\sigma$ on subsystem $A$ by measuring subsystem $B$.

\subsection{Two-way separable ENM states: optimal extraction}\label{app:two_way_optimal_success_prob}

The two-way $T$-type state given in the main text together with the following $F$-type state provide a concrete entry point to the general proof.

\begin{exmp}[Two-way $F$-type separable ENM state]\label{exmp:two-way-F-state}
For the face state $\ket{F}$ with Bloch vector $(1,1,1)/\sqrt3$, Theorem~\ref{thm:two_way_optimal} gives
\begin{equation}
\rho_F^{\leftrightarrow}\coloneqq(2-\sqrt{3})\ketbra{1}{1}^{\otimes 2}+\,\bigl(1-(2-\sqrt{3})\bigr)\ketbra{F}{F}^{\otimes 2}.
\end{equation}
Measuring either qubit in the $Z$ basis and postselecting the $0$ outcome prepares $\ket{F}$ on the other qubit with probability $1/\sqrt3$.
Thus, producing $N$ copies of $\ket{F}$ consumes, in expectation, $\sqrt3\,N$ copies of $\rho_F^{\leftrightarrow}$.
\end{exmp}

Here, mixing in the $\ketbra{1}{1}^{\otimes 2}$ branch places each marginal on the boundary of the stabilizer octahedron, while the $Z$-basis measurement filters this branch out.
We now show that the same geometric construction works for every pure single-qubit target and attains the optimal one-round extraction probability.

\subsubsection{Achievability}

\begin{proof}[Proof of achievability in Theorem~\ref{thm:two_way_optimal}]
Let $\mathbf s_\phi$ denote the Bloch vector of the stabilizer state $\ket{\phi}$.
Then
\begin{equation}
|\bracket{\phi}{\psi}|^2=\frac{1+\mathbf s_\phi\cdot\mathbf t}{2}.
\end{equation}
Since $\ket{\phi}$ maximizes this fidelity, we have
\begin{equation}
|\bracket{\phi}{\psi}|^2=\frac{1+\|\mathbf t\|_\infty}{2}.
\end{equation}
Let
\begin{equation}
p=\frac{2}{\|\mathbf t\|_1+1},\quad\Pi=\ketbra{\phi}{\phi}.
\end{equation}
The state $p\ketbra{\psi}{\psi}+(1-p)\ketbra{\phi^\perp}{\phi^\perp}$ has Bloch vector
\begin{equation}
\mathbf r=\frac{2\mathbf t-(\|\mathbf t\|_1-1)\mathbf s_\phi}{\|\mathbf t\|_1+1}.
\end{equation}
Since $\ket{\phi}$ maximizes the fidelity, there is a coordinate $j$ such that $|t_j|=\|\mathbf t\|_\infty$ and $\mathbf s_\phi=\operatorname{sgn}(t_j)\hat{\mathbf j}$.
For this coordinate,
\begin{equation}
r_j=\frac{\operatorname{sgn}(t_j)\bigl(2|t_j|-\|\mathbf t\|_1+1\bigr)}{\|\mathbf t\|_1+1}.
\end{equation}
Moreover,
\begin{equation}
2|t_j|-\|\mathbf t\|_1+1\geq1-\frac{\|\mathbf t\|_1}{3}>0,
\end{equation}
where we used $|t_j|\geq \|\mathbf t\|_1/3$ and $\|\mathbf t\|_1\leq \sqrt3$.
Hence the $j$-th component of $\mathbf r$ has the same sign as $t_j$.
For the remaining coordinates $\ell\neq j$, we have $r_\ell=2t_\ell/(\|\mathbf t\|_1+1)$.
Therefore
\begin{equation}
\|\mathbf r\|_1=\frac{2\sum_{\ell\neq j}|t_\ell|+\bigl(2|t_j|-\|\mathbf t\|_1+1\bigr)}{\|\mathbf t\|_1+1}=\frac{\|\mathbf t\|_1+1}{\|\mathbf t\|_1+1}=1.
\end{equation}
By the one-qubit stabilizer octahedron criterion, $p\ketbra{\psi}{\psi}+(1-p)\ketbra{\phi^\perp}{\phi^\perp}\in\stab_1$.
Moreover, $\Tr(\Pi\ketbra{\phi^\perp}{\phi^\perp})=0$.
Therefore Theorem~\ref{thm:two_way_general} shows that $\rho_\psi^\leftrightarrow$ is a two-way separable ENM state and that the local measurement $\Pi$ extracts $\ket{\psi}$ on either subsystem with success probability
\begin{equation}
p\,\Tr(\Pi\ketbra{\psi}{\psi})=\frac{2|\bracket{\phi}{\psi}|^2}{\|\mathbf t\|_1+1}=\frac{1+\|\mathbf t\|_\infty}{1+\|\mathbf t\|_1}.
\end{equation}
\end{proof}

\subsubsection{Optimality}

Let $\mathrm{CSP}_1$ denote the set of all CSP maps on 1 qubit.
The following lemma characterizes the structure of $\mathrm{CSP}_1$.
For one qubit, this class coincides with deterministic stabilizer protocols~\cite{Heimendahl2022axiomatic}.

\begin{lem}[Theorem~5 in~\cite{Heimendahl2022axiomatic}]\label{lemma:CPS1}
Let $\mc{E}$ be an extremal element of the convex set $\mathrm{CSP}_1$, then either $\mc{E}(\cdot)=U\cdot U^\dagger$ for some Clifford unitary $U$, or
\begin{equation}
\mc{E}(\cdot)=\sum_{s\in\{-1,1\}}U_s\Pi_{s}^P\cdot\Pi_{s}^P U_s^\dagger,
\end{equation}
where $U_{-1}, U_1$ are Clifford unitaries, $P \in \{X, Y, Z\}$ is a Pauli operator, and the projectors are defined as $\Pi_{s}^P \coloneqq \frac{\mathbb{I} + sP}{2}$.
\end{lem}

\begin{prop}\label{prop:two_way_optimality}
Let $\ket{\psi}$ be a one-qubit pure state, with Bloch vector $\mathbf t=(t_x,t_y,t_z)\in\mathbb{R}^3$ (i.e. $\ketbra{\psi}{\psi}=\frac{1}{2}\bigl(\mbb{I}+t_x X+t_y Y+t_z Z\bigr)$).
Suppose a $2$-qubit state $\rho$ on $A\otimes B$ satisfies
\begin{enumerate}
\item permutation symmetric, i.e., $\mathrm{SWAP}\cdot\rho\cdot\mathrm{SWAP}=\rho$;
\item separable;
\item the reduced density matrix $\rho_A\in\stab_1$ (which implies $\rho_B\in\stab_1$ since $\rho_A=\rho_B$);
\item there exist $\mc{P}_A,\mc{P}_B\in\mathrm{CSP}_1$, such that the post-measurement state
\begin{equation}\label{eq:postselect_psi}
\bra{0}\otimes\mbb{I}\,\mc{P}_A\otimes\mc{P}_B(\rho)\,\ket{0}\otimes\mbb{I}\propto\ketbra{\psi}{\psi}.
\end{equation}
\end{enumerate}
Then the success probability satisfies
\begin{equation}\label{eq:succ_bound_general}
\Tr(\bra{0}\otimes\mbb{I}\,\mc{P}_A\otimes\mc{P}_B(\rho)\,\ket{0}\otimes\mbb{I})\le \frac{1+\|\mathbf t\|_\infty}{1+\|\mathbf t\|_1},
\end{equation}
where $\|\mathbf t\|_1:=|t_x|+|t_y|+|t_z|$ and $\|\mathbf t\|_\infty:=\max\{|t_x|,|t_y|,|t_z|\}$.
\end{prop}

Here $\mc{P}_B$ should be understood as the allowed local stabilizer post-processing on the unmeasured output qubit.
Equivalently, condition~4 can be written as
\begin{equation*}
\mc{P}_B\!\left(\bra{0}\otimes\mbb{I}\,\mc{P}_A\otimes\mc{I}(\rho)\,\ket{0}\otimes\mbb{I}\right)\propto\ketbra{\psi}{\psi}.
\end{equation*}
We place $\mc{P}_B$ before the projection in \eqref{eq:postselect_psi} only as a compact notation.

\begin{proof}[Proof of Proposition~\ref{prop:two_way_optimality}]
If $\ket{\psi}$ is a stabilizer state, then $\|\mathbf t\|_\infty=\|\mathbf t\|_1=1$, and the right-hand side of \eqref{eq:succ_bound_general} is $1$.
The claim is then trivial.
Hence we assume below that $\ket{\psi}\notin\stab_1$.

\textbf{Claim.} If there exists a two-qubit state $\rho$ and stabilizer protocol $\mc{P}_A\otimes \mc{P}_B$ satisfying 1--4, then there exist a 2-qubit state $\rho'$ satisfying 1--3, and a 1-qubit stabilizer state $\ket{S^*}$, such that
\begin{equation}
\bra{S^*}\otimes\mbb{I}\,\rho'\,\ket{S^*}\otimes\mbb{I}\propto\ketbra{\psi}{\psi},
\end{equation}
and has a measurement success probability
\begin{equation}
\Tr(\bra{S^*}\otimes\mbb{I}\,\rho'\,\ket{S^*}\otimes\mbb{I})\ge\Tr(\bra{0}\otimes\mbb{I}\,\mc{P}_A\otimes\mc{P}_B(\rho)\,\ket{0}\otimes\mbb{I}).
\end{equation}

\begin{proof}[Proof of claim]
Since $\mc{P}_B$ is trace-preserving, we know
\begin{equation}
\frac{\bra{0}\otimes\mbb{I}\,\mc{P}_A\otimes\mc{P}_B(\rho)\,\ket{0}\otimes\mbb{I}}{\Tr(\bra{0}\otimes\mbb{I}\,\mc{P}_A\otimes\mc{P}_B(\rho)\,\ket{0}\otimes\mbb{I})}=\mc{P}_B\left(\sigma\right),
\end{equation}
where
\begin{equation}
\sigma\coloneqq\frac{\bra{0}\otimes\mbb{I}\,\mc{P}_A\otimes\mc{I}(\rho)\,\ket{0}\otimes\mbb{I}}{\Tr(\bra{0}\otimes\mbb{I}\,\mc{P}_A\otimes\mc{I}(\rho)\,\ket{0}\otimes\mbb{I})}
\end{equation}
is a single qubit state.
Hence $\mc{P}_B(\sigma)=\ketbra{\psi}{\psi}$.

By Lemma~\ref{lemma:CPS1}, we can write $\mc{P}_B(\sigma)$ as
\begin{equation}
\mc{P}_B(\sigma)=\sum_{i=1}^{N_U}p_i\,U^{(i)}\sigma U^{(i)\dagger}+\sum_{i=N_U+1}^{N_U+N_P}p_i\sum_{s\in\{-1,1\}}U^{(i)}_s\Pi_s^{P^{(i)}}\sigma\Pi_s^{P^{(i)}} U^{(i)\dagger}_s.
\end{equation}
Here, $\{p_i\}_{i=1}^{N_U+N_P}$ is a probability distribution with $p_i>0$ for all $i$; $U^{(i)}$, $U^{(i)}_s$ are Clifford unitaries; $P^{(i)}$'s are Pauli operators; $\Pi_s^{P^{(i)}}=\frac{\mbb{I}+sP^{(i)}}{2}$ are projectors.
$\ketbra{\psi}{\psi}$ is a pure state, thus an extremal point in the convex set of 1-qubit density matrices.
Since $U^{(i)}\sigma U^{(i)\dagger}$ and $\sum_{s\in\{-1,1\}}U^{(i)}_s\Pi_s^{P^{(i)}}\sigma\Pi_s^{P^{(i)}} U^{(i)\dagger}_s$ are all valid density matrices, and their convex combination equals  $\ketbra{\psi}{\psi}$, we know all of them must equal $\ketbra{\psi}{\psi}$.
Suppose there are no unitary terms ($N_U=0$). Then $\mc{P}_B(\sigma)\in\stab$, which is not possible since $\ketbra{\psi}{\psi}\notin\stab$.
Therefore $N_U>0$, and $U^{(1)}\sigma U^{(1)\dagger}=\ketbra{\psi}{\psi}$.
This means $\sigma$ is a pure state.
We take
\begin{equation}
\rho'\coloneqq \big(U^{(1)}\otimes U^{(1)}\big)\rho \big(U^{(1)\dagger}\otimes U^{(1)\dagger}\big),
\end{equation}
which satisfies 1--3 in the proposition's statement.

Again we write
\begin{equation}
\mc{P}_A(\cdot)=\sum_{i=1}^{M_U}q_i\,V^{(i)}\cdot V^{(i)\dagger}+\sum_{i=M_U+1}^{M_U+M_P}q_i\sum_{s\in\{-1,1\}}V^{(i)}_s\Pi_s^{Q^{(i)}}\cdot\Pi_s^{Q^{(i)}}V^{(i)\dagger}_s .
\end{equation}
For any 1-qubit stabilizer state $\ket{S}$, define
\begin{equation}
R_S:=\bra{S}\otimes\mbb{I}\,\rho\,\ket{S}\otimes\mbb{I}.
\end{equation}
Let $\ket{S_i}:=V^{(i)\dagger}\ket{0}$, and let $\ket{Q_s^{(i)}}$ be the stabilizer state satisfying $\Pi_s^{Q^{(i)}}=\ketbra{Q_s^{(i)}}{Q_s^{(i)}}$.
Set
\begin{equation}
c_{i,s}:=\abs{\bra{0}V_s^{(i)}\ket{Q_s^{(i)}}}^2\in[0,1].
\end{equation}
Then
\begin{equation}\label{eq:2_qubit_optimal_PA_decomposition}
\bra{0}\otimes\mbb{I}\,\mc{P}_A\otimes\mc{I}(\rho)\,\ket{0}\otimes\mbb{I}=\sum_{i=1}^{M_U}q_i R_{S_i}+\sum_{i=M_U+1}^{M_U+M_P}q_i\sum_{s\in\{-1,1\}}c_{i,s}R_{Q_s^{(i)}}.
\end{equation}
Since the left-hand side is proportional to $\sigma$, and $\sigma$ is pure, every nonzero positive summand on the right-hand side is proportional to $\sigma$.

Define
\begin{equation}
M:=
\max_{\substack{\ket{S}\in\stab_1\\ R_S\neq0,\;R_S\propto\sigma}}
\Tr(R_S).
\end{equation}

For each unitary branch, the contribution to the trace is at most $q_iM$.
For each measurement branch, the two projectors $\ketbra{Q_+^{(i)}}{Q_+^{(i)}}$ and $\ketbra{Q_-^{(i)}}{Q_-^{(i)}}$ form a Pauli measurement, so
\begin{equation}
R_{Q_+^{(i)}}+R_{Q_-^{(i)}}=\rho_B.
\end{equation}
At most one of the two weighted terms $c_{i,s}R_{Q_s^{(i)}}$ can be nonzero.
Indeed, if both were nonzero, then both $R_{Q_+^{(i)}}$ and $R_{Q_-^{(i)}}$ would be proportional to $\sigma$, and hence $\rho_B\propto\sigma$, contradicting $\rho_B\in\stab_1$ and $\sigma\notin\stab_1$.
Thus each measurement branch also contributes at most $q_iM$ to the trace.
Therefore
\begin{equation}
\Tr\!\left(\bra{0}\otimes\mbb{I}\,\mc{P}_A\otimes\mc{I}(\rho)\,\ket{0}\otimes\mbb{I}\right)\le\sum_{i=1}^{M_U+M_P}q_iM=M.
\end{equation}
Choose a stabilizer state $\ket{S^{*}_{\mathrm{temp}}}$ attaining this maximum $M$.
Then
\begin{equation}
R_{S^{*}_{\mathrm{temp}}}\propto\sigma,\quad\Tr(R_{S^{*}_{\mathrm{temp}}})\ge\Tr\!\left(\bra{0}\otimes\mbb{I}\,\mc{P}_A\otimes\mc{I}(\rho)\,\ket{0}\otimes\mbb{I}\right).
\end{equation}

Now take
\begin{equation}
\ket{S^*}\coloneqq U^{(1)}\ket{S^{*}_{\mathrm{temp}}},
\end{equation}
we can conclude that
\begin{align}
(\bra{S^*}\otimes\mbb{I})\,\rho'\,(\ket{S^*}\otimes\mbb{I})=&(\bra{S^{*}_{\mathrm{temp}}}\otimes U^{(1)})\,\rho\,(\ket{S^{*}_{\mathrm{temp}}}\otimes U^{(1)\dagger})\\
={}&U^{(1)}\Big[(\bra{S^{*}_{\mathrm{temp}}}\otimes\mbb{I})\,\rho\,(\ket{S^{*}_{\mathrm{temp}}}\otimes\mbb{I})\Big]U^{(1)\dagger}\\
\propto{}&U^{(1)}\sigma U^{(1)\dagger}=\ketbra{\psi}{\psi}.
\end{align}
and the success probability
\begin{align}
\Tr(\bra{S^*}\otimes\mbb{I}\,\rho'\,\ket{S^*}\otimes\mbb{I})=&\Tr(\bra{S^{*}_{\mathrm{temp}}}\otimes\mbb{I}\,\rho\,\ket{S^{*}_{\mathrm{temp}}}\otimes\mbb{I})\\
\ge&\Tr(\bra{0}\otimes\mbb{I}\,\mc{P}_A\otimes\mc{P}_B(\rho)\,\ket{0}\otimes\mbb{I}),
\end{align}
where the last line uses that $\mc{P}_B$ is trace-preserving.
\end{proof}

Therefore, we know the maximum success probability is upper bounded by
\begin{equation}\label{eq:succ_le_max_f}
\Tr(\bra{0}\otimes\mbb{I}\,\mc{P}_A\otimes\mc{P}_B(\rho)\,\ket{0}\otimes\mbb{I})\le\max_{\ketbra{S}{S}\in\stab} f_\psi(\ket{S}),
\end{equation}
where the maximum is taken over $1$-qubit stabilizer states $\{\ket{0},\ket{1},\ket{+},\ket{-},\ket{y+},\ket{y-}\}$, and
\begin{align}
f_\psi(\ket{S})\coloneqq\textbf{max}_{\rho}~ &\Tr(\bra{S}\otimes\mbb{I}\,\rho\,\ket{S}\otimes\mbb{I})\\
\textbf{ s.t. } & \rho\ge0,\Tr(\rho)=1,\\
&\mathrm{SWAP}\cdot\rho\cdot\mathrm{SWAP} = \rho,\\
& \Tr_B(\rho)\in\stab_1,\\
&\rho^{\mathrm{T}_B}\ge0,\\
& \bra{S}\otimes\mbb{I}\,\rho\,\ket{S}\otimes\mbb{I}= \ketbra{\psi}{\psi}\Tr(\bra{S}\otimes\mbb{I}\,\rho\,\ket{S}\otimes\mbb{I}).\label{eq:postselect_constraint_psi}
\end{align}
This is an SDP that admits an analytical solution.

\begin{lem}[Closed form of $\max_{\ketbra{S}{S}}f_\psi(\ket{S})$]\label{lem:fpsi_closed_form}
Let $\ket{\psi}\notin\stab_1$ be pure with Bloch vector $\mathbf t=(t_x,t_y,t_z)$.
Let $\ket{S}$ be a stabilizer state with Bloch vector $\mathbf s=(s_x,s_y,s_z)$.
Let $k\in\{x,y,z\}$ be the unique nonzero coordinate index of $\mathbf s$.
Then
\begin{equation}\label{eq:fpsi_piecewise}
f_\psi(\ket{S})=
\begin{cases}
0, & \displaystyle\sum_{j\neq k}|t_j|>1+s_k t_k,\\[0.8em]
\dfrac{1+s_k t_k}{1+s_k t_k+\sum_{j\neq k}|t_j|}, & \displaystyle\sum_{j\neq k}|t_j|\le 1+s_k t_k.
\end{cases}
\end{equation}
This implies
\begin{equation}\label{eq:fpsi_max_closed}
\max_{\ketbra{S}{S}\in\stab} f_\psi(\ket{S})=\dfrac{1+\|\mathbf t\|_\infty}{1+\|\mathbf t\|_1}.
\end{equation}
\end{lem}

\begin{proof}
We solve the SDP defining $f_\psi(\ket{S})$.

Let $\sigma_B:=\bra{S}\otimes \mbb{I}\,\rho\,\ket{S}\otimes \mbb{I}$.
Constraint \eqref{eq:postselect_constraint_psi} makes $\sigma_B$ rank one and proportional to $\ketbra{\psi}{\psi}$.
Since this is a two-qubit optimization, the PPT constraint is equivalent to separability.
Together with the SWAP symmetry, this allows us to write a symmetric separable decomposition
\begin{equation}
\rho=\sum_\ell p_\ell\,\frac{\ket{a_\ell b_\ell}\!\bra{a_\ell b_\ell}+\ket{b_\ell a_\ell}\!\bra{b_\ell a_\ell}}{2},\quad p_\ell\ge0,\ \sum_\ell p_\ell=1.
\end{equation}
Then
\begin{equation}
\sigma_B=\sum_\ell p_\ell\,\frac{|\langle S|a_\ell\rangle|^2\ketbra{b_\ell}{b_\ell}+|\langle S|b_\ell\rangle|^2\ketbra{a_\ell}{a_\ell}}{2}.
\end{equation}
A positive mixture is rank one only if every pure state appearing with nonzero weight is identical to $\ket{\psi}$.
Since $\ket{\psi}\notin\stab_1$, each term in the above decomposition is therefore of one of two types: either it contributes, in which case $\ketbra{a_\ell}{a_\ell}=\ketbra{b_\ell}{b_\ell}=\ketbra{\psi}{\psi}$, or it does not contribute, in which case $\langle S|a_\ell\rangle=\langle S|b_\ell\rangle=0$, i.e. $\ketbra{a_\ell}{a_\ell}=\ketbra{b_\ell}{b_\ell}=\ketbra{S^\perp}{S^\perp}$, where $\ket{S^\perp}$ is the stabilizer state orthogonal to $\ket{S}$.
Therefore any feasible $\rho$ must be of the form
\begin{equation}\label{eq:rho_w}
\rho(w)=w\,\ketbra{\psi\psi}{\psi\psi}+(1-w)\,\ketbra{S^\perp S^\perp}{S^\perp S^\perp},\quad 0\le w\le1.
\end{equation}
The objective equals
\begin{equation}\label{eq:obj_w}
\Tr\!\bigl[(\ketbra{S}{S}\otimes \mbb{I})\rho(w)\bigr]=w\,|\langle S|\psi\rangle|^2=w\,\frac{1+s_k t_k}{2}.
\end{equation}
The last equality uses $|\langle S|\psi\rangle|^2=(1+\mathbf s\cdot\mathbf t)/2=(1+s_k t_k)/2$.
The partial trace of $\rho(w)$ over $B$ yields
\begin{equation}
\Tr_B\rho(w)=w\,\ketbra{\psi}{\psi}+(1-w)\,\ketbra{S^\perp}{S^\perp}.
\end{equation}
Its Bloch vector is
\begin{equation}
\mathbf r(w)=w\,\mathbf t+(1-w)(-\mathbf s).
\end{equation}
The condition $\Tr_B(\rho)\in\stab_1$ is equivalent to $\|\mathbf r(w)\|_1\le1$.
Since $s_k\in\{-1,+1\}$, one checks
\begin{equation}
\|\mathbf r(w)\|_1=\bigl| -1+w(1+s_k t_k)\bigr|+w\sum_{j\neq k}|t_j|.
\end{equation}
Since $\ket{\psi}\notin\stab_1$, we have $1+s_k t_k>0$.
The absolute value changes sign at $w=(1+s_k t_k)^{-1}$.
Hence the stabilizer constraint is equivalent to
\begin{equation}
\begin{cases}
1-w(1+s_k t_k)+w\sum_{j\neq k}|t_j|\le1,
& 0\le w\le (1+s_k t_k)^{-1},\\[0.4em]
-1+w(1+s_k t_k)+w\sum_{j\neq k}|t_j|\le1,
& w\ge (1+s_k t_k)^{-1}.
\end{cases}
\end{equation}

First consider $\sum_{j\neq k}|t_j|>1+s_k t_k$.
In the first branch, the inequality is equivalent to $w\left(\sum_{j\neq k}|t_j|-(1+s_k t_k)\right)\le0$, thus no $w>0$ is feasible.
The second branch has no feasible $w$: its lower endpoint already violates the constraint, and its left-hand side increases with $w$.
Therefore $w=0$, and $f_\psi(\ket{S})=0$.

Then consider $\sum_{j\neq k}|t_j|\le1+s_k t_k$.
The first branch is always satisfied.
The second branch gives
\begin{equation}
w\le \frac{2}{1+s_k t_k+\sum_{j\neq k}|t_j|}.
\end{equation}
This bound is at least $(1+s_k t_k)^{-1}$ because $\sum_{j\neq k}|t_j|\le1+s_k t_k$, so it is compatible with the second branch.
Moreover, this upper bound is below $1$: indeed, the condition $\sum_{j\neq k}|t_j|\le1+s_k t_k$ forces $s_k t_k=|t_k|$, since otherwise it would imply $\|\mathbf t\|_1\le1$, contradicting $\ket{\psi}\notin\stab_1$; hence the denominator is $1+\|\mathbf t\|_1>2$.
Taking also $0\le w\le1$ into account, the full feasible interval is
\begin{equation}
0\le w\le \frac{2}{1+s_k t_k+\sum_{j\neq k}|t_j|}.
\end{equation}
The objective in \eqref{eq:obj_w} is increasing in $w$, so the optimum is attained at the right endpoint.
Hence
\begin{equation}
f_\psi(\ket{S})=\frac{1+s_k t_k}{1+s_k t_k+\sum_{j\neq k}|t_j|},
\end{equation}
which proves \eqref{eq:fpsi_piecewise}.

Choose $k^*\in\{x,y,z\}$ such that $|t_{k^*}|=\|\mathbf t\|_\infty$, and take the stabilizer state $\ket{S}$ whose Bloch vector has nonzero coordinate $k^*$ with sign $\mathrm{sgn}(t_{k^*})$.
Then $s_{k^*}t_{k^*}=\|\mathbf t\|_\infty$ and $\sum_{j\neq k^*}|t_j|=\|\mathbf t\|_1-\|\mathbf t\|_\infty$, so \eqref{eq:fpsi_piecewise} gives
\begin{equation}
f_\psi(\ket{S})=\frac{1+\|\mathbf t\|_\infty}{1+\|\mathbf t\|_1}.
\end{equation}
Conversely, for any coordinate $k$, the opposite sign gives zero since $\sum_{j\neq k}|t_j|>1-|t_k|$ for $\ket{\psi}\notin\stab_1$.
For the aligned sign, \eqref{eq:fpsi_piecewise} gives
\begin{equation}
f_\psi(\ket{S})=\frac{1+|t_k|}{1+\|\mathbf t\|_1}\le \frac{1+\|\mathbf t\|_\infty}{1+\|\mathbf t\|_1}.
\end{equation}
Hence \eqref{eq:fpsi_max_closed} follows.
\end{proof}

Combining \eqref{eq:succ_le_max_f} and Lemma~\ref{lem:fpsi_closed_form} yields
\begin{equation}
\Tr(\bra{0}\otimes\mbb{I}\,\mc{P}_A\otimes\mc{P}_B(\rho)\,\ket{0}\otimes\mbb{I})\le \frac{1+\|\mathbf t\|_\infty}{1+\|\mathbf t\|_1},
\end{equation}
which proves Proposition~\ref{prop:two_way_optimality}.
\end{proof}

\subsection{The golden ENM state: extremality, magic localization and distillation}\label{app:golden-details}

\paragraph{Exact RoM and stabilizer-noise robustness.}
In the Pauli basis, we have
\begin{equation}
\rho_{\mathrm{golden}}=\frac{1}{4}\Big[\mbb{I}\otimes\mbb{I}+\frac{1}{\sqrt{5}}\bigl(X\otimes\mbb{I}+Y\otimes\mbb{I}+Z\otimes X+Z\otimes Y+Z\otimes Z\bigr)\Big].
\label{eq:golden-pauli}
\end{equation}
Let
\begin{equation}
A_{\mathrm{golden}}\coloneqq X\otimes\mbb{I}+Y\otimes\mbb{I}+Z\otimes X+Z\otimes Y+Z\otimes Z .
\end{equation}
The five Pauli operators appearing in $A_{\mathrm{golden}}$ pairwise anticommute.
For any pure two-qubit stabilizer state, at most one of these five Pauli operators can belong, up to sign, to its stabilizer group; all the remaining expectations are zero.
Hence $\abs{\Tr(\phi A_{\mathrm{golden}})}\le1$ for every pure stabilizer state $\phi$, and by convexity the same holds for every $\phi\in\stab_2$.
Therefore $A_{\mathrm{golden}}$ is feasible for the dual RoM problem \eqref{eq:RoM_dual}, and
\begin{equation}
\mc{R}(\rho_{\mathrm{golden}})\ge\Tr(\rho_{\mathrm{golden}}A_{\mathrm{golden}})=\sqrt{5}.
\end{equation}
The opposite inequality follows from the known two-qubit upper bound $\mc{R}(\rho)\le\sqrt{5}$ for all two-qubit states~\cite{Howard2017application}.
Thus $\mc{R}(\rho_{\mathrm{golden}})=\sqrt{5}$.

Equivalently, the minimum amount of stabilizer noise needed to wash out the magic of $\rho_{\mathrm{golden}}$ is
\begin{equation}
\begin{aligned}
s_{\mathrm{golden}}&\coloneqq\min\Big\{s\,\Big|\,\frac{\rho_{\mathrm{golden}}+s\omega}{1+s}\in\stab,\ \omega\in\stab\Big\}\\
&=\frac{\mc{R}(\rho_{\mathrm{golden}})-1}{2}=\frac{\sqrt{5}-1}{2}=\varphi^{-1},
\end{aligned}
\label{eq:golden-noise}
\end{equation}
where $\varphi=(1+\sqrt{5})/2$ is the golden ratio.

\paragraph{Magic localization and distillation.}

The golden state also has an operational interpretation.
Measuring the first qubit of $\rho_\mathrm{golden}$ in the $Z$ basis gives outcomes $0$ and $1$ with equal probability.
Conditioned on the outcome, the state's magic is localized onto the second qubit, whose state is
\begin{equation}
\sigma(\tfrac{1}{\sqrt{5}},\tfrac{1}{\sqrt{5}},\tfrac{1}{\sqrt{5}})\quad\text{or}\quad\sigma(-\tfrac{1}{\sqrt{5}},-\tfrac{1}{\sqrt{5}},-\tfrac{1}{\sqrt{5}}),
\end{equation}
respectively.
If the outcome is ignored, the second qubit remains the stabilizer marginal $(\rho_\mathrm{golden})_B=\mbb{I}/2$.
Similarly, measuring the second qubit in the $Z$ basis localizes the state's magic onto the first qubit, whose state is
\begin{equation}
\sigma(\tfrac{1}{\sqrt{5}},\tfrac{1}{\sqrt{5}},\tfrac{1}{\sqrt{5}})\quad\text{or}\quad\sigma(\tfrac{1}{\sqrt{5}},\tfrac{1}{\sqrt{5}},-\tfrac{1}{\sqrt{5}}),
\end{equation}
while the unconditional state on the first qubit is the stabilizer state $(\rho_{\mathrm{golden}})_A=\sigma(\tfrac{1}{\sqrt{5}},\tfrac{1}{\sqrt{5}},0)$.
Thus a local stabilizer measurement plus a classical outcome unlocks a nonstabilizer single-qubit state on the other side.

All these conditional states are Clifford-equivalent noisy face states.
Indeed,
\begin{equation}
\sigma(\tfrac{1}{\sqrt{5}},\tfrac{1}{\sqrt{5}},\tfrac{1}{\sqrt{5}})=\sigma(\tfrac{r_{\mathrm g}}{\sqrt{3}},\tfrac{r_{\mathrm g}}{\sqrt{3}},\tfrac{r_{\mathrm g}}{\sqrt{3}}),\quad r_{\mathrm g}=\sqrt{\frac{3}{5}}.
\end{equation}
$r_{\mathrm g}=\sqrt{3/5}$ is above the Bravyi--Kitaev five-qubit distillation threshold $r_{\mathrm{th}}=\sqrt{3/7}$ for face-type magic states~\cite{Bravyi2005universal,rall2017signedquantumweightenumerators}.
Consequently, after the measurement outcome is communicated and a Clifford correction is applied to align the signs, sufficiently many copies of $\rho_{\mathrm{golden}}$ can be converted by local stabilizer operations and classical communication into arbitrarily high-fidelity one-qubit face magic states.
These distilled states can then be consumed by standard magic-state injection to implement a non-Clifford gate.
Without the classical measurement outcomes, however, each party only sees a stabilizer marginal and has no local access to the hidden magic.

\section{Separable ENM channels: Choi characterization and reversible embedding}\label{app:BipartiteChannels}

\subsection{Proof of Fact~\ref{fact:enm-channel-choi}}\label{app:enm-channel-choi}

By \cite[Lemma~4.2]{Seddon2019quantifying}, a CPTP map $\mc{G}$ is CSP if and only if $\Lambda_{\mc{G}}\in\stab$.
The reduced maps satisfy
\begin{align}
\Lambda_{\mc{E}_A}&=\Tr_{B'B}(\Lambda_{\mc{E}}),\\
\Lambda_{\mc{E}_B}&=\Tr_{A'A}(\Lambda_{\mc{E}}).
\end{align}
Therefore $\mc{E}_A,\mc{E}_B\in\mathrm{CSP}$ if and only if the two marginals of $\Lambda_{\mc{E}}$ on the partition $A'A:B'B$ are stabilizer states, which is precisely the locally magic-free condition for $\Lambda_{\mc{E}}$.
The same Choi criterion gives $\mc{E}\notin\mathrm{CSP}$ if and only if $\Lambda_{\mc{E}}\notin\stab$, so $\mc{E}$ is ENM if and only if $\Lambda_{\mc{E}}$ is ENM.

By \cite[Proposition~6.22]{watrous2018theory}, the CPTP map $\mc{E}$ is separable if and only if $\Lambda_{\mc{E}}$ is separable with respect to the same partition.
Combining these equivalences proves the claim.

\subsection{Proof of Theorem~\ref{thm:separableENMchannel}}\label{app:proof-separable-enm-channel}

For a unitary $U$, write $[U](\rho)=U\rho U^\dagger$.
Let $\mc{E}\notin\mathrm{CSP}$ be an $n$-qubit channel, and suppose probabilities $\{p_i\}_{i=0}^{M-1}$ and Clifford unitaries $\{W_i\}_{i=0}^{M-1}$ satisfy $\sum_{i=0}^{M-1}p_i[W_i]\circ\mc{E}\in\mathrm{CSP}$.
Set $m=\lceil\log_2M\rceil$ and $X(i)\coloneqq\bigotimes_{j=1}^mX^{i_j}$, where $(i_1,\cdots,i_m)$ is the $m$-bit expansion of $i$.
The channel under consideration is
\begin{equation}\mc{F}\coloneqq \sum_{i=0}^{M-1}p_i[X(i)]\otimes\big([W_i]\circ\mc{E}\big).\end{equation}
Each summand in this expression is a product channel and the coefficients $\{p_i\}$ form a probability distribution, so $\mc{F}$ is LOSR and hence separable.
Its reduced channels on the label register $R$ and system register $S$ are
\begin{align}
\mc{F}_R&=\sum_{i=0}^{M-1}p_i[X(i)],\\
\mc{F}_S&=\sum_{i=0}^{M-1}p_i[W_i]\circ\mc{E}.
\end{align}
The first is a convex mixture of Pauli channels and hence CSP, while the second is CSP by hypothesis.
Thus $\mc{F}$ is locally magic-free.

We first construct the stabilizer superchannel $\Theta_1$.
When $\Theta_1(\mc{E})$ is applied to an input state $\rho_{RS}$ of an $m$-qubit label register $R$ and an $n$-qubit system $S$, the protocol is:
\begin{enumerate}
\item Classically sample $i\in\{0,\cdots,M-1\}$ according to $\{p_i\}$.
\item Apply $X(i)$ to $R$.
\item Apply $\mc{E}$ to $S$, followed by the Clifford unitary $W_i$.
\end{enumerate}
This gives $\Theta_1(\mc{E})=\mc{F}$.
The protocol uses only classical randomness and Clifford unitaries, so $\Theta_1$ is a stabilizer superchannel.

To apply the decoding superchannel $\Theta_2(\mc{F})$ to an input state $\rho_S$, proceed as follows:
\begin{enumerate}
\item Prepare $R$ in $\ket{0}^{\otimes m}$.
\item Apply $\mc{F}$ to $\ketbra{0}{0}^{\otimes m}_R\otimes\rho_S$.
\item Measure $R$ in the computational basis and, conditioned on outcome $i$, apply $W_i^\dagger$ to $S$.
\item Discard $R$.
\end{enumerate}
This gives $\Theta_2(\mc{F})=\mc{E}$.
Stabilizer-state preparation, computational-basis measurement, classical control and Clifford unitaries are all stabilizer operations, so $\Theta_2$ is a stabilizer superchannel.

If $\mc{F}$ were CSP, then the stabilizer superchannel $\Theta_2$ would imply $\mc{E}=\Theta_2(\mc{F})\in\mathrm{CSP}$, contradicting the hypothesis.
Hence $\mc{F}$ is globally non-CSP and therefore ENM.

\section{Activation key protocol}\label{app:ActivationKey}

\subsection{Activation key protocol as a separable ENM channel}\label{app:activation key-single-gate-decoder}

Let $K$ and $S$ denote the key and system registers, respectively, and set $T^\perp\coloneqq ZT$.
The $T$ gate activation key protocol is described by
\begin{equation}\mc{F}_T\coloneqq\frac{1}{2}[\mbb{I}_K\otimes T_S]+\frac{1}{2}[X_K\otimes T^\perp_S].\end{equation}
Each branch is a product channel, and the two branches are coordinated only by shared classical randomness.
Hence $\mc{F}_T$ is LOSR and therefore separable.
$\mc{F}_T$ is locally magic-free.

To show that the joint channel nevertheless contains recoverable magic, let $\mc{G}$ be a channel acting on $K\otimes S$, and let $\rho$ be a state on $S$.
Define the superchannel
\begin{align}
\Theta_{\mathrm{dec}}(\mc{G})(\rho)\coloneqq\sum_{b=0,1}Z_S^b\operatorname{Tr}_K\Big[\big(\ket{b}\!\bra{b}_K\otimes\mbb{I}_S\big)\mc{G}\big(\ket{0}\!\bra{0}_K\otimes\rho\big)\Big]Z_S^b.
\end{align}
It prepares $K$ in $\ket{0}$, measures the output key in the computational basis, applies $Z^b$ conditioned on outcome $b$, and discards $K$.
All these pre- and post-processing steps are stabilizer operations, so $\Theta_{\mathrm{dec}}$ is a stabilizer superchannel.
Direct substitution gives
\begin{equation}
\Theta_{\mathrm{dec}}(\mc{F}_T)=[T].\label{eq:activation-decoding-T}
\end{equation}
Operationally, withholding the key leaves the user with complete dephasing, whereas communicating $b$ and applying $Z^b$ activates the intended $T$ gate.

\subsection{Hiding common non-Clifford gates}\label{app:activation key-gate-hiding}

We now verify the Clifford-randomization condition of Theorem~\ref{thm:separableENMchannel} for several standard magic gates beyond the single-$T$ construction above.
We will repeatedly use the following criterion for diagonal channels~\cite{Seddon2019quantifying}.

\begin{lem}
\label{lem:diagonal-channel-phase-state}
Let $\mc D$ be an $n$-qubit channel whose Kraus operators are diagonal in the computational basis.
Then $\mc D\in\mathrm{CSP}$ if and only if
\begin{equation}
\mc D(\ketbra{+}{+}^{\otimes n})\in\stab_n .
\end{equation}
\end{lem}

\begin{proof}
Write $K_\alpha=\sum_x k_\alpha(x)\ketbra{x}{x}$ for the diagonal Kraus operators of $\mc D$, and set $c_{xy}\coloneqq\sum_\alpha k_\alpha(x)\overline{k_\alpha(y)}$.
Then
\begin{align}
\mc D(\ketbra{+}{+}^{\otimes n})&=2^{-n}\sum_{x,y}c_{xy}\ketbra{x}{y},\\
\Lambda_{\mc D}&=2^{-n}\sum_{x,y}c_{xy}\ketbra{xx}{yy}=V\,\mc D(\ketbra{+}{+}^{\otimes n})\,V^\dagger,
\end{align}
where $V\ket{x}=\ket{x}\ket{x}$.
The isometry $V$ is a stabilizer isometry: it is obtained by appending $\ket{0}^{\otimes n}$ and applying CNOT gates.
Hence $V\rho V^\dagger$ is a stabilizer state if and only if $\rho$ is.
By \cite[Lemma~4.2]{Seddon2019quantifying}, a channel is CSP if and only if its Choi state is stabilizer, which proves the claim.
\end{proof}

\paragraph{$T\otimes T$.}

Define the reduced two-qubit channel
\begin{equation}
\mc D_{T,T}\coloneqq\frac{1}{2}[T^{\otimes 2}]+\frac{1}{2}[(T^\perp)^{\otimes 2}].
\end{equation}
This is a diagonal channel.
By Lemma~\ref{lem:diagonal-channel-phase-state}, it suffices to evaluate its phase state:
\begin{equation}
\mc D_{T,T}\!\left(\ketbra{+}{+}^{\otimes 2}\right)
=
\frac{1}{2}\ketbra{T}{T}^{\otimes 2}
+
\frac{1}{2}\ketbra{T^\perp}{T^\perp}^{\otimes 2}
=
\frac{1}{4}\begin{pmatrix}
1&0&0&-i\\
0&1&1&0\\
0&1&1&0\\
i&0&0&1
\end{pmatrix},
\end{equation}
which is a stabilizer state, since
\begin{equation}
\frac{1}{4}\begin{pmatrix}
1&0&0&-i\\
0&1&1&0\\
0&1&1&0\\
i&0&0&1
\end{pmatrix}
=\frac{1}{2}\ketbra{\Phi_i}{\Phi_i}+\frac{1}{2}\ketbra{\Psi^+}{\Psi^+},
\label{eq:multipartite-T-two-body}
\end{equation}
where $\ket{\Phi_i}\coloneqq(\ket{00}+i\ket{11})/\sqrt{2}$ and $\ket{\Psi^+}\coloneqq(\ket{01}+\ket{10})/\sqrt{2}$ are stabilizer states.
Hence $\mc D_{T,T}\in\mathrm{CSP}$.
The corresponding key-augmented separable ENM channel is
\begin{equation}
\mc{F}_{T,T}\coloneqq\frac{1}{2}[\mbb{I}\otimes T^{\otimes 2}]+\frac{1}{2}[X\otimes(T^\perp)^{\otimes 2}].
\end{equation}

\paragraph{$\mathrm{CC}Z$ gates.}

For the three-qubit $\mathrm{CC}Z$ gate, one local $Z$ correction on any participating qubit is enough.
Let
\begin{equation}
\mc D_{\mathrm{CC}Z}\coloneqq\frac{1}{2}[\mathrm{CC}Z]+\frac{1}{2}[(Z\otimes\mbb{I}\otimes\mbb{I})\mathrm{CC}Z]=[\mathrm{CC}Z]\circ\frac{1}{2}\big([\mbb{I}^{\otimes 3}]+[Z\otimes\mbb{I}\otimes\mbb{I}]\big).
\end{equation}
This is diagonal, and its phase state is
\begin{align}
\mc D_{\mathrm{CC}Z}\!\left(\ketbra{+}{+}^{\otimes 3}\right)&=[\mathrm{CC}Z]\left(\frac{\mbb{I}}{2}\otimes\ketbra{+}{+}^{\otimes 2}\right)\\
&=\frac{1}{2}\ketbra{0}{0}\otimes\ketbra{+}{+}^{\otimes 2}+\frac{1}{2}\ketbra{1}{1}\otimes[\mathrm{C}Z]\left(\ketbra{+}{+}^{\otimes 2}\right).
\end{align}
Both branches are stabilizer states; hence the phase state is stabilizer and $\mc D_{\mathrm{CC}Z}\in\mathrm{CSP}$ by Lemma~\ref{lem:diagonal-channel-phase-state}.
Therefore Theorem~\ref{thm:separableENMchannel} yields
\begin{equation}
\mc{F}_{\mathrm{CC}Z}\coloneqq\frac{1}{2}[\mbb{I}\otimes\mathrm{CC}Z]+\frac{1}{2}[X\otimes((Z\otimes\mbb{I}\otimes\mbb{I})\mathrm{CC}Z)].
\end{equation}

\paragraph{$\mathrm{C}S$ gates.}

$\mathrm{C}S=\diag(1,1,1,i)$.
Let
\begin{equation}
\mc D_{\mathrm{C}S}\coloneqq\frac{1}{2}[\mathrm{C}S]+\frac{1}{2}[(Z\otimes\mbb{I})\mathrm{C}S]=[\mathrm{C}S]\circ\frac{1}{2}\big([\mbb{I}^{\otimes 2}]+[Z\otimes\mbb{I}]\big).
\end{equation}
Then
\begin{equation}
\mc D_{\mathrm{C}S}\!\left(\ketbra{+}{+}^{\otimes 2}\right)=\frac{1}{2}\ketbra{0}{0}\otimes\ketbra{+}{+}+\frac{1}{2}\ketbra{1}{1}\otimes\ketbra{+_y}{+_y},
\end{equation}
where $\ket{+_y}\coloneqq S\ket{+}=(\ket{0}+i\ket{1})/\sqrt{2}$.
This is a stabilizer state, so $\mc D_{\mathrm{C}S}\in\mathrm{CSP}$.
The corresponding separable ENM channel is
\begin{equation}
\mc F_{\mathrm{C}S}\coloneqq\frac{1}{2}[\mbb{I}\otimes\mathrm{C}S]+\frac{1}{2}[X\otimes((Z\otimes\mbb{I})\mathrm{C}S)].
\end{equation}

\paragraph{Controlled-phase rotations.}

Consider an arbitrary controlled-phase rotation $\mathrm{C}R(\theta)\coloneqq\diag(1,1,1,e^{i\pi\theta})$.
Such gates appear in the standard circuit decomposition of the quantum Fourier transform and hence in Shor's algorithm~\cite{Shor1994Algorithms}.

In this case we use the Clifford correction $\mathrm{C}Z$ and define
\begin{equation}
\mc D_{\mathrm{C}R(\theta)}\coloneqq\frac{1}{2}[\mathrm{C}R(\theta)]+\frac{1}{2}[\mathrm{C}Z\cdot\mathrm{C}R(\theta)]=[\mathrm{C}R(\theta)]\circ\frac{1}{2}\big([\mbb{I}^{\otimes 2}]+[\mathrm{C}Z]\big).
\end{equation}
The phase state is independent of $\theta$:
\begin{align}
\mc D_{\mathrm{C}R(\theta)}\!\left(\ketbra{+}{+}^{\otimes 2}\right)&=[\mathrm{C}R(\theta)]\left(\frac{1}{2}\ketbra{++}{++}+\frac{1}{2}[\mathrm{C}Z](\ketbra{++}{++})\right)\\
&=\frac{1}{2}\ketbra{++}{++}+\frac{1}{2}[\mathrm{C}Z](\ketbra{++}{++}),
\end{align}
because the averaged state has no coherences between $\ket{11}$ and the other computational-basis states.
It is a convex combination of stabilizer states, so $\mc D_{\mathrm{C}R(\theta)}\in\mathrm{CSP}$ by Lemma~\ref{lem:diagonal-channel-phase-state}.
Whenever $[\mathrm{C}R(\theta)]\notin\mathrm{CSP}$, Theorem~\ref{thm:separableENMchannel} therefore gives
\begin{equation}
\mc F_{\mathrm{C}R(\theta)}\coloneqq\frac{1}{2}[\mbb{I}\otimes\mathrm{C}R(\theta)]+\frac{1}{2}[X\otimes\mathrm{C}Z\cdot\mathrm{C}R(\theta)]
\end{equation}
as a separable ENM channel for $\mathrm{C}R(\theta)$.

\paragraph{Pauli rotation gates.}

Let $P$ be an $n$-qubit Pauli string and
\begin{equation}
R_P(\theta)\coloneqq e^{-i\theta P/2}=\cos(\theta/2)\mbb{I}-i\sin(\theta/2)P.
\end{equation}
Pauli rotations include common parametrized gates used in variational circuits, such as QAOA layers~\cite{Cerezo2021variational,Bharti2022NISQ}.
Averaging over the Pauli correction $P$ gives
\begin{align}
\mc D_{R_P(\theta)}&\coloneqq\frac{1}{2}[R_P(\theta)]+\frac{1}{2}[P\cdot R_P(\theta)]\\
&=\frac{1}{2}\big([\mbb{I}]+[P]\big)\circ[R_P(\theta)]=\frac{1}{2}\big([\mbb{I}]+[P]\big),
\end{align}
where the last equality holds because $R_P(\theta)$ acts trivially after dephasing in the eigenspaces of $P$.
The right-hand side is a Pauli dephasing channel and hence belongs to $\mathrm{CSP}$.
Whenever $[R_P(\theta)]\notin\mathrm{CSP}$, Theorem~\ref{thm:separableENMchannel} gives the separable ENM channel
\begin{equation}
\mc F_{R_P(\theta)}\coloneqq\frac{1}{2}[\mbb{I}\otimes R_P(\theta)]+\frac{1}{2}[X\otimes P\cdot R_P(\theta)].
\end{equation}

\subsection{Correlated gate masking in layered circuits}

The next proposition formalizes an elementary closure property of CSP channels.
In a layered circuit, suppose that all but a selected set of local gates are fixed CSP channels, while the selected gates are correlated through a common classical random variable.
Importantly, these selected gates may appear in different layers and may even be causally connected through the circuit.
If the averaged joint channel on these selected locations is CSP, then the entire circuit is CSP.
Thus, to prove that a randomized circuit is magic-free, it suffices to check the correlated block of randomized gates; the surrounding CSP gates cannot create magic.

\begin{figure}[t]
\centering
\includegraphics[width=0.67\textwidth]{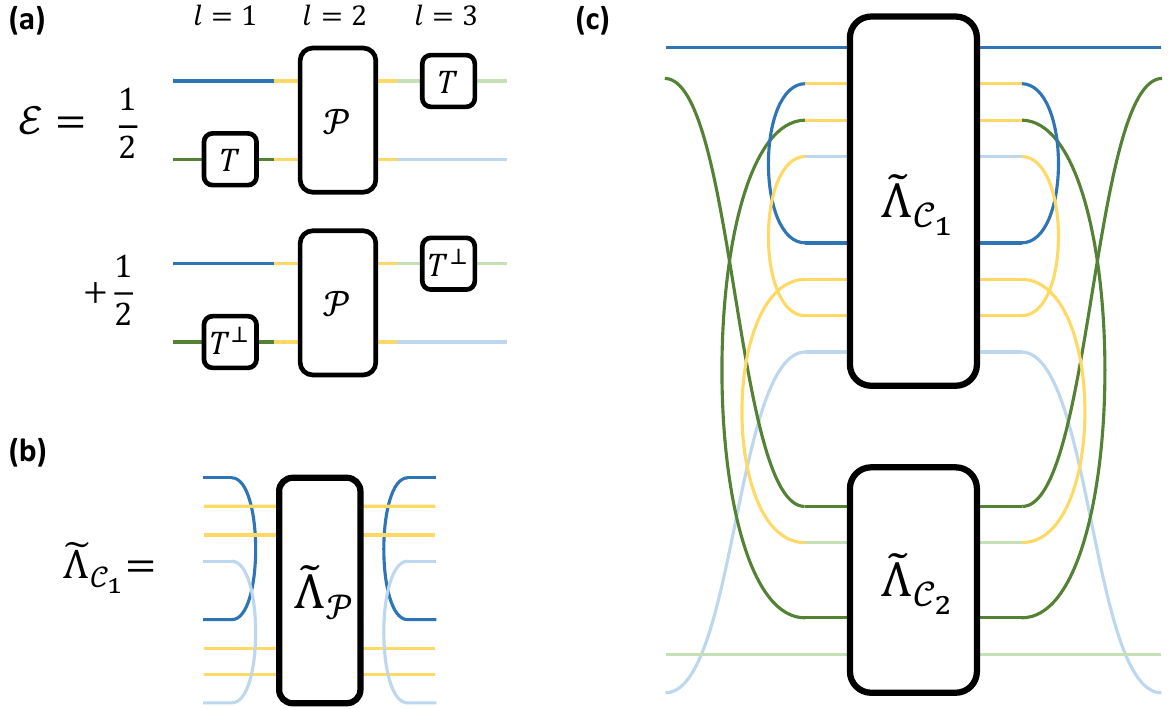}
\caption{Graphical proof of Proposition~\ref{prop:Choi_stab_thus_no_capacity}, for the case $\mc{E}=\frac{1}{2}([T]\otimes\mc{I})\circ\mc{P}\circ(\mc{I}\otimes[T])+\frac{1}{2}([T^\perp]\otimes\mc{I})\circ\mc{P}\circ(\mc{I}\otimes[T^\perp])$, where $\mc{P}\in\mathrm{CSP}$.
Here $N=2$, $L=3$, $K=2$; $\{p_i\}_{i=1}^N=\{\frac{1}{2},\frac{1}{2}\}$; $\mc{E}^{(1)}_{1,1}=\mc{E}^{(2)}_{1,1}=\mc{I}$, $\mc{E}^{(1)}_{1,2}=[T]$, $\mc{E}^{(1)}_{2,1}=\mc{E}^{(2)}_{2,1}=\mc{P}$, $\mc{E}^{(1)}_{3,1}=[T]$, $\mc{E}^{(1)}_{3,2}=\mc{E}^{(2)}_{3,2}=\mc{I}$, $\mc{E}^{(2)}_{1,2}=[T^\perp]$, $\mc{E}^{(2)}_{3,1}=[T^\perp]$.
For a CPTP map $\mc{F}$, we let $\widetilde{\Lambda}_{\mc{F}}
\coloneqq
(\mc{I}\otimes\mc{F})
\bigl(\ket{\widetilde{\Phi}^+}\!\bra{\widetilde{\Phi}^+}\bigr)$ where $\ket{\widetilde{\Phi}^+}=\sum_{i}\ket{ii}$ be the unnormalized Choi matrix of $\mc{F}$.
We have $\Lambda_\mc{F}=\widetilde{\Lambda}_{\mc{F}}/\Tr(\widetilde{\Lambda}_{\mc{F}})$.
(a) shows the circuit structure of $\mc{E}$.
(b) is the tensor network representation of the unnormalized Choi matrix $\widetilde{\Lambda}_{\mc{C}_1}$ of $\mc{C}_1=\mc{I}\otimes\mc{P}\otimes\mc{I}$.
$\widetilde{\Lambda}_{\mc{C}_2}$ can be drawn similarly, where $\mc{C}_2=\frac{1}{2}[T]\otimes[T]+\frac{1}{2}[T^\perp]\otimes[T^\perp]\in\mathrm{CSP}$.
(c) is the unmormailized Choi matrix $\widetilde{\Lambda}_{\mc{E}}$ of $\mc{E}$, from which we know how to obtian $\Lambda_{\mc{E}}$ by stabilizer operations on $\Lambda_{\mc{C}_1}\otimes\Lambda_{\mc{C}_2}$.}
\label{fig:Choi_stab}
\end{figure}

\begin{prop}\label{prop:Choi_stab_thus_no_capacity}
Let $\{p_i\}_{i=1}^N$ be a probability distribution.
Consider a channel of the form
\begin{equation}
\mc{E}=\sum_{i=1}^Np_i\,\Big(\bigotimes_{j=1}^{m_L}\mc{E}^{(i)}_{L,j}\Big)\circ\cdots\circ\Big(\bigotimes_{j=1}^{m_1}\mc{E}^{(i)}_{1,j}\Big),
\end{equation}
where each layer $l$ is a tensor product of local CPTP maps $\mc{E}^{(i)}_{l,j}$ acting on disjoint subsystems.
For a fixed position $(l,j)$, all $\mc{E}^{(1)}_{l,j},\cdots,\mc{E}^{(N)}_{l,j}$ act on the same subsystem.
Suppose there exists a set $\{(l_k,j_k)\}_{k=1}^K
\subseteq\{(l,j):1\le l\le L,\ 1\le j\le m_l\}$ such that
\begin{enumerate}
\item For all $(l,j)\notin\{(l_k,j_k)\}_{k=1}^K$, we have $\mc{E}^{(1)}_{l,j}=\cdots=\mc{E}^{(N)}_{l,j}=:\mc{E}_{l,j}\in\mathrm{CSP}$.
\item The $K$-partite channel $\sum_{i=1}^N p_i\,\mc{E}_{l_1,j_1}^{(i)} \otimes \cdots \otimes \mc{E}_{l_K,j_K}^{(i)}\in\mathrm{CSP}$.
\end{enumerate}
Then the overall channel $\mc{E}$ is CSP.
\end{prop}

\begin{proof}
By \cite[Lemma 4.2]{Seddon2019quantifying}, a CPTP map $\mc F$ is CSP iff its Choi state $\Lambda_{\mc F}$ is a stabilizer state.
Thus it suffices to show that $\Lambda_{\mc E}\in\mathrm{STAB}$.

Denote $\mc{C}_1\coloneqq\bigotimes_{(l,j)\notin\{(l_k,j_k)\}_{k=1}^K}\mc{E}_{l,j}$ and $\mc{C}_2=\sum_{i=1}^N p_i\,\mc{E}_{l_1,j_1}^{(i)} \otimes \cdots \otimes \mc{E}_{l_K,j_K}^{(i)}$.
Since $\Lambda_{\mc{C}_1}$ and $\Lambda_{\mc{C}_2}$ are stabilizer states, and $\Lambda_{\mc{E}}$ can be obtained from $\Lambda_{\mc{C}_1}\otimes\Lambda_{\mc{C}_2}$ by
\begin{enumerate}
\item swapping several qubits,
\item measuring several qubit pairs in the Bell basis and post-selecting in the outcome $\ket{\Phi^+}=\frac{1}{\sqrt{2}}(\ket{00}+\ket{11})$,
\end{enumerate}
we know $\Lambda_{\mc{E}}$ must be a stabilizer state.
See Fig.~\ref{fig:Choi_stab} for an example.

Therefore $\mc{E}\in\mathrm{CSP}$.
\end{proof}

\section{Magic scaling of multipartite separable ENM states}\label{app:multipartite-scaling}

We now consider the many-body magic of the following two families of
separable states, which are ENM for $n\ge2$ and $n\ge3$, respectively:
\begin{equation}
\begin{aligned}
\rho_\psi^{(n)}&\coloneqq \frac{\|\mathbf t\|_1-1}{\|\mathbf t\|_1+1}\ketbra{\phi^\perp}{\phi^\perp}^{\otimes n}+\frac{2}{\|\mathbf t\|_1+1}\ketbra{\psi}{\psi}^{\otimes n},\\
\rho_T^{(n)}&\coloneqq\frac{1}{2}\ketbra{T}{T}^{\otimes n}+\frac{1}{2}\ketbra{T^\perp}{T^\perp}^{\otimes n}.
\end{aligned}
\end{equation}
For an $n$-qubit state $\rho$ and $\alpha>0$, $\alpha\neq1$, For an $n$-qubit state $\rho$, we use the following Pauli-moment
extension of the stabilizer Rényi entropy (SRE)~\cite{Lorenzo2022stabilizer} to mixed states:
\begin{equation}
M_\alpha(\rho)\coloneqq \frac{1}{1-\alpha}\log_2\frac{1}{2^n}\sum\nolimits_{P\in\mc{P}_n^+}\abs{\Tr(P\rho)}^{2\alpha},
\end{equation}
where $\mc{P}_n^+$ is the set of $n$-qubit Pauli operators with phase $+1$.
For $1/2\le\alpha<1$, it satisfies
\begin{equation}\label{eq:mixed-SRE-RoM-bound}
M_\alpha(\rho)\le2\log_2\mc{R}(\rho).
\end{equation}
Because SRE does not require solving an optimization problem, it is generally more tractable~\cite{PhysRevB.107.035148,Lami2023Nonstabilizerness,Tarabunga2024MPS} and admits analytical expressions in the cases below.

The calculations below show that, for fixed $1/2\le\alpha<1$, both families obey the common asymptotic scaling law
\begin{equation}
M_\alpha(\rho^{(n)})=nM_\alpha(\ketbra{\chi}{\chi})+C_{\alpha,\chi}+o(1),
\end{equation}
where $(\rho^{(n)},\ket{\chi})$ denotes either $(\rho_\psi^{(n)},\ket{\psi})$ or $(\rho_T^{(n)},\ket{T})$, and $C_{\alpha,\chi}$ is independent of $n$.
Thus each mixed $n$-qubit state has the same extensive coefficient as the corresponding underlying one-qubit pure state, with mixing contributing only a constant.

\subsection{Scaling of $\rho_\psi^{(n)}$}

Let $\ket{\psi}$ be a pure one-qubit nonstabilizer state with Bloch vector $\mathbf t=(t_x,t_y,t_z)$ and RoM $\|\mathbf t\|_1$.
Choose $k\in\{x,y,z\}$ such that $|t_k|=\|\mathbf t\|_\infty$, and let $\ket{\phi}$ be the stabilizer state whose Bloch vector is $\operatorname{sgn}(t_k)\hat{\mathbf k}$.
Write
\begin{equation}
p\coloneqq \frac{2}{\|\mathbf t\|_1+1},\quad q\coloneqq \frac{\|\mathbf t\|_1-1}{\|\mathbf t\|_1+1}=1-p,
\end{equation}
so that
\begin{equation}
\rho_\psi^{(n)}=q\ketbra{\phi^\perp}{\phi^\perp}^{\otimes n}+p\ketbra{\psi}{\psi}^{\otimes n}.
\end{equation}
For $P\in\mc{P}_n^+$, let $n_a(P)$ denote the number of sites on which $P$ acts as $a\in\{X,Y,Z\}$.
Tensor-product factorization gives $\Tr(P\ketbra{\psi}{\psi}^{\otimes n})=t_x^{n_x(P)}t_y^{n_y(P)}t_z^{n_z(P)}$.
The Bloch vector of $\ket{\phi^\perp}$ is $-\operatorname{sgn}(t_k)\hat{\mathbf k}$.
Consequently, $\Tr(P\ketbra{\phi^\perp}{\phi^\perp}^{\otimes n})$ vanishes if $P$ contains a Pauli operator along any direction $\ell\neq k$, whereas it equals $[-\operatorname{sgn}(t_k)]^{n_k(P)}$ if every site carries either $I$ or the Pauli operator along $k$.
Combining the two product-state expectations with weights $p$ and $q$, and using $t_k=\operatorname{sgn}(t_k)|t_k|$ in the latter case, gives
\begin{equation}
\Tr(P\rho_\psi^{(n)})=
\begin{cases}
p\,t_x^{n_x(P)}t_y^{n_y(P)}t_z^{n_z(P)},
& n_\ell(P)>0\text{ for some }\ell\neq k,\\[2mm]
\bigl(\operatorname{sgn}(t_k)\bigr)^{n_k(P)}
\left[p|t_k|^{n_k(P)}+q(-1)^{n_k(P)}\right],
& n_\ell(P)=0\text{ for all }\ell\neq k.
\end{cases}
\end{equation}

Fix $1/2\le\alpha<1$ and define
\begin{equation}
A_\alpha\coloneqq 1+\sum_{a=x,y,z}|t_a|^{2\alpha},\quad B_\alpha\coloneqq 1+|t_k|^{2\alpha}.
\end{equation}
We evaluate the contributions from the two cases separately.

In the first case, the $\ket{\phi^\perp}$ term vanishes and each Pauli string contributes $p^{2\alpha}\prod_{a=x,y,z}|t_a|^{2\alpha n_a(P)}$ to the Pauli moment.
If we first sum over all Pauli strings without imposing the condition ``$n_\ell(P)>0\text{ for some }\ell\neq k$'', we have
\begin{equation}
\sum_{P\in\mc{P}_n^+}\prod_{a=x,y,z}|t_a|^{2\alpha n_a(P)}=\left(1+\sum_{a=x,y,z}|t_a|^{2\alpha}\right)^n=A_\alpha^n.
\end{equation}
The complement of the condition ``$n_\ell(P)>0$ for some $\ell\neq k$'' is ``$n_\ell(P)=0$ for every $\ell\neq k$''.
For a string in this complement, every $P_j$ is either $I$ or the Pauli operator along $k$, so
\begin{equation}
\sum_{\substack{P\in\mc{P}_n^+\\n_\ell(P)=0\text{ for all }\ell\neq k}}\prod_{a=x,y,z}|t_a|^{2\alpha n_a(P)}=\left(1+|t_k|^{2\alpha}\right)^n=B_\alpha^n.
\end{equation}
Subtracting this complementary contribution from the unrestricted sum gives
\begin{align}
&p^{2\alpha}\sum_{\substack{P\in\mc{P}_n^+\\n_\ell(P)>0\text{ for some }\ell\neq k}}\prod_{a=x,y,z}|t_a|^{2\alpha n_a(P)}\\
=&p^{2\alpha}\left[\left(1+\sum_{a=x,y,z}|t_a|^{2\alpha}\right)^n-\left(1+|t_k|^{2\alpha}\right)^n\right]\\
=&p^{2\alpha}\left(A_\alpha^n-B_\alpha^n\right).
\end{align}

In the second case, every site carries either $I$ or the Pauli operator along $k$.
For fixed $m\in\{0,\ldots,n\}$, a string with $n_k(P)=m$ is uniquely determined by choosing the $m$ sites that carry the Pauli operator along $k$; all remaining sites carry $I$.
There are therefore exactly $\binom{n}{m}$ such strings:
\begin{equation}
\sum_{\substack{P\in\mc{P}_n^+\\
n_k(P)=m\\
n_\ell(P)=0\text{ for all }\ell\neq k}}1
=\binom{n}{m}.
\end{equation}
Since $|t_k|=\|\mathbf t\|_\infty>0$, we have $\operatorname{sgn}(t_k)\in\{\pm1\}$.
For all strings with the same $m$, the sign prefactor in the expectation value has unit modulus and hence disappears after taking the absolute value:
\begin{align}
\abs{\Tr(P\rho_\psi^{(n)})}^{2\alpha}
&=\abs{\operatorname{sgn}(t_k)}^{2\alpha m}
\left|p|t_k|^m+q(-1)^m\right|^{2\alpha}\\
&=\left|p|t_k|^m+q(-1)^m\right|^{2\alpha}.
\end{align}
Multiplying this common value by $\binom{n}{m}$ and summing over $m$ gives the contribution of the second case:
\begin{equation}
\sum_{m=0}^{n}\binom{n}{m}\left|p|t_k|^m+q(-1)^m\right|^{2\alpha}.
\end{equation}
The endpoint $m=0$ includes the identity string.
The two cases are disjoint and together exhaust $\mc{P}_n^+$, so adding their contributions yields the exact Pauli moment
\begin{equation}
\sum_{P\in\mc{P}_n^+}\abs{\Tr(P\rho_\psi^{(n)})}^{2\alpha}=p^{2\alpha}\left(A_\alpha^n-B_\alpha^n\right)+\sum_{m=0}^{n}\binom{n}{m}\left|p\,|t_k|^{m}+q(-1)^m\right|^{2\alpha}.\label{eq:SRE-rhopsi-exact}
\end{equation}
We now identify the leading exponential term.
Because $\ket{\psi}$ is pure and nonstabilizer, $\sum_a|t_a|^2=1$ and at least two components are nonzero.
Since $\alpha<1$, we have
\begin{equation}
\sum_{a=x,y,z}|t_a|^{2\alpha}>\sum_{a=x,y,z}|t_a|^2=1.
\end{equation}
Moreover, at least one component away from $k$ is nonzero, so $A_\alpha>2$ and $A_\alpha>B_\alpha$.

Since $2\alpha\ge1$, the inequality\footnote{Set $r=2\alpha\ge1$. The triangle inequality gives $|u+v|^r\le(|u|+|v|)^r$, while convexity of $x\mapsto x^r$ on $[0,\infty)$ gives $(a+b)^r\le2^{r-1}(a^r+b^r)$. Taking $a=|u|$ and $b=|v|$ proves the stated bound.} $|u+v|^{2\alpha}\le2^{2\alpha-1}(|u|^{2\alpha}+|v|^{2\alpha})$ gives
\begin{equation}
\sum_{m=0}^{n}\binom{n}{m}\left|p|t_k|^m+q(-1)^m\right|^{2\alpha}\le 2^{2\alpha-1}\left[p^{2\alpha}\sum_{m=0}^n\binom{n}{m}|t_k|^{2\alpha m}+q^{2\alpha}\sum_{m=0}^n\binom{n}{m}\right]=2^{2\alpha-1}\left[p^{2\alpha}B_\alpha^n+q^{2\alpha}2^n\right].
\end{equation}
Both $B_\alpha$ and $2$ are strictly smaller than $A_\alpha$.
Setting $r_\alpha\coloneqq \max\{B_\alpha,2\}/A_\alpha<1$, Eq.~\eqref{eq:SRE-rhopsi-exact} therefore yields
\begin{equation}
\sum_{P\in\mc{P}_n^+}\abs{\Tr(P\rho_\psi^{(n)})}^{2\alpha}=p^{2\alpha}A_\alpha^n\left[1+O(r_\alpha^n)\right].
\end{equation}
Therefore,
\begin{align}
M_\alpha(\rho_\psi^{(n)})
&=\frac{1}{1-\alpha}\log_2\left\{\frac{1}{2^n}p^{2\alpha}A_\alpha^n\left[1+O(r_\alpha^n)\right]\right\}\\
&=\frac{\log_2 A_\alpha-1}{1-\alpha}\,n
+\frac{2\alpha}{1-\alpha}\log_2p
+\frac{1}{1-\alpha}\log_2\left[1+O(r_\alpha^n)\right].
\end{align}
Since $r_\alpha<1$ and $\alpha$ is fixed, $\log_2[1+O(r_\alpha^n)]=O(r_\alpha^n)$.
Moreover,
\begin{align}
M_\alpha(\ketbra{\psi}{\psi})
&=\frac{1}{1-\alpha}\log_2\left[\frac{1}{2}\left(1+\sum_{a=x,y,z}|t_a|^{2\alpha}\right)\right]\\
&=\frac{\log_2 A_\alpha-1}{1-\alpha}.
\end{align}
Recall $p=2/(\|\mathbf t\|_1+1)$, we therefore obtain
\begin{equation}
M_\alpha(\rho_\psi^{(n)})
=n\,M_\alpha(\ketbra{\psi}{\psi})
+\frac{2\alpha}{1-\alpha}\log_2\frac{2}{\|\mathbf t\|_1+1}
+O(r_\alpha^n).
\end{equation}
Thus, although every one-qubit marginal of $\rho_\psi^{(n)}$ is stabilizer, its Pauli-moment functional $M_\alpha$ is extensive and has the same leading coefficient as that of $\ket{\psi}^{\otimes n}$; mixing changes only the constant term, up to exponentially small corrections.

\subsection{Scaling of $\rho_T^{(n)}$}

For a Pauli string $P\in\mc{P}_n^+$, let
\begin{equation}
w_{xy}(P)\coloneqq \#\{\text{sites with }X\text{ or }Y\},\quad w_{z}(P)\coloneqq \#\{\text{sites with }Z\}.
\end{equation}
Since the $X$ and $Y$ expectation values of $\ket{T}$ and $\ket{T^\perp}$ are $\pm1/\sqrt{2}$, while their $Z$ expectation values vanish, we have
\begin{equation}
\Tr(P\rho_T^{(n)})
=
\begin{cases}
2^{-w_{xy}(P)/2},&\text{ when }w_z(P)=0\text{ and }w_{xy}(P)\text{ even}.\\
0,&\text{ otherwise}.
\end{cases}
\end{equation}
Therefore,
\begin{align}
\sum_{P\in\mc{P}_{n}^{+}}
\big|\Tr(P\rho_T^{(n)})\big|^{2\alpha}
&=\sum_{\substack{k=0\\ k\text{ even}}}^{n}
\binom{n}{k}2^{k}\bigl(2^{-k}\bigr)^{\alpha}
=\sum_{k=0}^{n}
\binom{n}{k}\frac{1+(-1)^k}{2}2^{k(1-\alpha)}\\
&=\tfrac12\big[(1+2^{1-\alpha})^{n}+(1-2^{1-\alpha})^{n}\big].
\end{align}
This implies
\begin{equation}
M_{\alpha}(\rho_T^{(n)})=\frac{1}{1-\alpha}\Big[\log_{2}\big[(1+2^{1-\alpha})^{n}+(1-2^{1-\alpha})^{n}\big]-(n+1)\Big].
\end{equation}
Since
\begin{equation}
(1+2^{1-\alpha})^{n}+(1-2^{1-\alpha})^{n}=(1+2^{1-\alpha})^{n}\Big[1+\Big(\frac{1-2^{1-\alpha}}{1+2^{1-\alpha}}\Big)^n\Big],
\end{equation}
Notice that $r\coloneqq \abs{\frac{1-2^{1-\alpha}}{1+2^{1-\alpha}}}<1$ for all $\alpha$.
Therefore, we have
\begin{align}
M_{\alpha}(\rho_T^{(n)})=&\frac{1}{1-\alpha}\Big[n\log_{2}(1+2^{1-\alpha})+\log_2\Big[1+\Big(\frac{1-2^{1-\alpha}}{1+2^{1-\alpha}}\Big)^n\Big]-(n+1)\Big]\\
=&\frac{\log_{2}(1+2^{1-\alpha})-1}{1-\alpha}n-\frac{1}{1-\alpha}+\frac{1}{1-\alpha}\log_2\Big[1+\Big(\frac{1-2^{1-\alpha}}{1+2^{1-\alpha}}\Big)^n\Big]\\
=&\frac{\log_{2}(1+2^{1-\alpha})-1}{1-\alpha}n-\frac{1}{1-\alpha}+O\left(r^n\right).
\end{align}
The leading coefficient is positive and equals the single-copy SRE,
\begin{equation}
M_{\alpha}(\ketbra{T}{T})=\frac{\log_{2}(1+2^{1-\alpha})-1}{1-\alpha}>0.
\end{equation}
Hence
\begin{equation}
M_{\alpha}(\rho_T^{(n)})=nM_\alpha(\ketbra{T}{T})-\frac{1}{1-\alpha}+O\left(r^n\right).
\end{equation}

\section{Experimental characterization of ENM states and protocols}\label{app:experiments}

\subsection{Activation key experiment: complete adjacent-pair process tomography}
\label{app:activation key-tomography}

To complement the two representative induced channels shown in Fig.~\ref{fig:activationkey}(c), we report the reconstructed channels for every adjacent pair in the 10-qubit chain.
Fig.~\ref{fig:activation key-all-channels} shows the real and imaginary parts of the process matrix for each of the nine pairs.
$\text{L}\mathcal{R}$ represents the log-robustness of magic, all values of which are close to 0.

\begin{figure}[h!]
\centering
\includegraphics[width=0.98\textwidth]{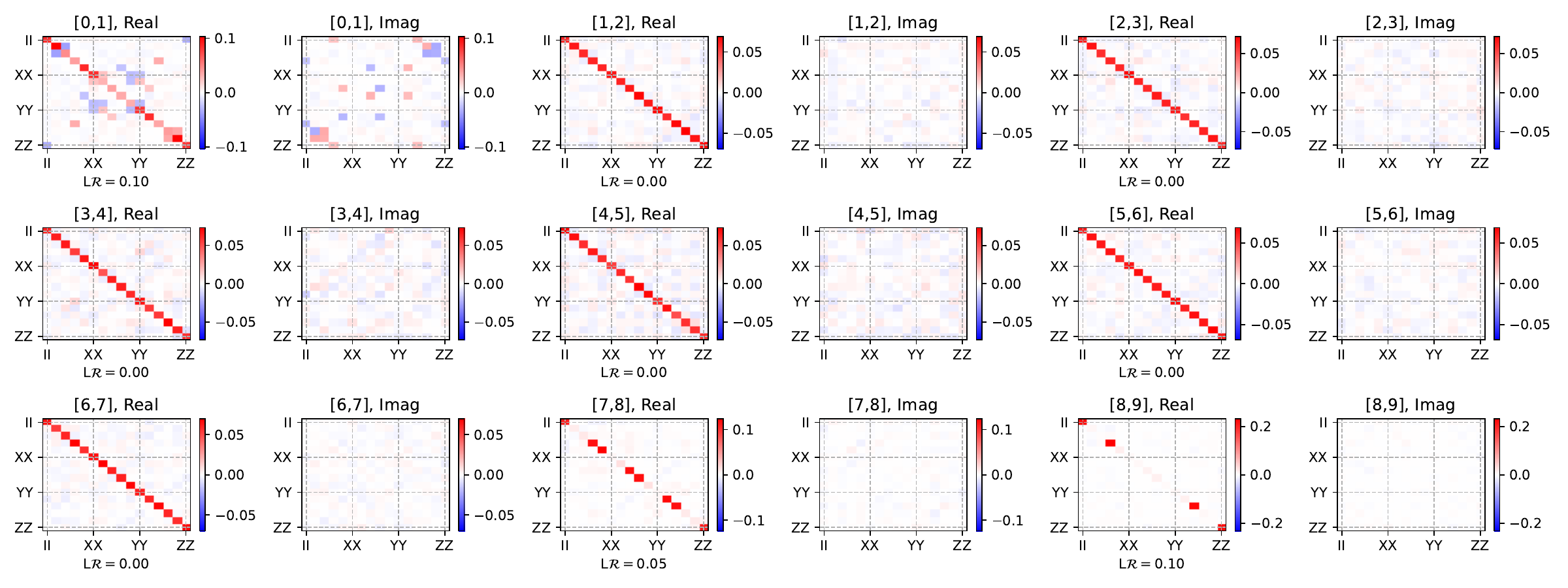}
\caption{Complete experimental adjacent-pair process tomography for the locked 10-qubit circuit in Fig.~\ref{fig:activationkey}.
For each pair $[j,j+1]$, the left and right heat maps show the real and imaginary parts of the reconstructed process matrix $\chi$, respectively.
$\text{L}\mathcal{R}$ indicates the log-robustness of magic.}
\label{fig:activation key-all-channels}
\end{figure}

\subsection{Separable ENM state experiments}\label{app:state-experiments}

\subsubsection{Two-way $T$- and $F$-type ENM states}

We prepare the two-way $T$- and $F$-type separable ENM states and conditionally extract the target magic state on either subsystem by measuring the other subsystem.
Across the five physical-qubit pairs (Ua,Da), (Ub,Db), (Uc,Dc), (Ud,Dd), and (Ue,De), the reconstructed ENM-state fidelities range from approximately $0.982$ to $0.992$.
The conditional extraction fidelities range from approximately $0.974$ to $0.993$ for $\ket{T}$ and from approximately $0.970$ to $1.000$ for $\ket{F}$.
Fig.~\ref{fig:two-way-state-experiment} shows the preparation circuits, physical-qubit layouts, and complete bidirectional-extraction data.

\begin{figure}[h!]
\centering
\includegraphics[width=0.92\textwidth]{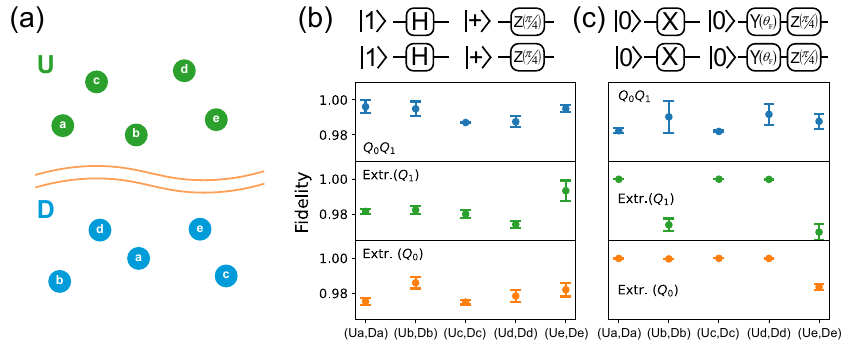}
\caption{
Experimental demonstration of separable ENM states on a multi-qubit superconducting processor.
(a) On a multiqubit chip, we selected 5 qubits from the upper half and 5 qubits from the lower half to experimentally generate the two-way T- and F-type separable entirely nonlocal magic (ENM) states.
The selected pairs are spatially separated and need not be directly connected on the device, illustrating that their classically correlated preparation does not rely on hardware connectivity.
(b) and (c) Experimental results for the T-type and F-type ENM states, respectively.
The blue data points represent the measured fidelity of the joint two-qubit mixed state ($Q_0 Q_1$).
The green (orange) data points show the fidelity of the conditional postmeasurement state on $Q_1$ ($Q_0$) with the corresponding ideal $\ket{T}$ or $\ket{F}$ state after measuring $Q_0$ ($Q_1$).
For the $T$-type state, extraction is conditioned on the $+1$ outcome of an $X$-basis measurement, implemented by applying $H$ before computational-basis readout and retaining outcome $0$; for the $F$-type state, extraction is conditioned on outcome $0$ of a $Z$-basis measurement.
Error bars denote statistical uncertainties.}
\label{fig:two-way-state-experiment}
\end{figure}

\subsubsection{Golden ENM state}

We prepare $\rho_{\mathrm{golden}}=\frac{1}{4}I\otimes I+\frac{1}{4\sqrt{5}}(X\otimes I+Y\otimes I+Z\otimes X+Z\otimes Y+Z\otimes Z)$ by sampling two product-state preparation circuits with equal probability, so the state is separable by construction.
Over ten repeated data sets, with each reported point obtained by averaging five runs, the reconstructed-state fidelity lies between approximately $0.993$ and $0.998$.
The reconstructed global logarithmic RoM lies between approximately $0.790$ and $0.802$, close to the two-qubit optimum $\ln{\sqrt{5}}\simeq0.805$, while the logarithmic RoM of each single-qubit marginal is reconstructed at the stabilizer value $0$; see Fig.~\ref{fig:golden-state-experiment}.

\begin{figure}[h!]
\centering
\includegraphics[width=0.60\textwidth]{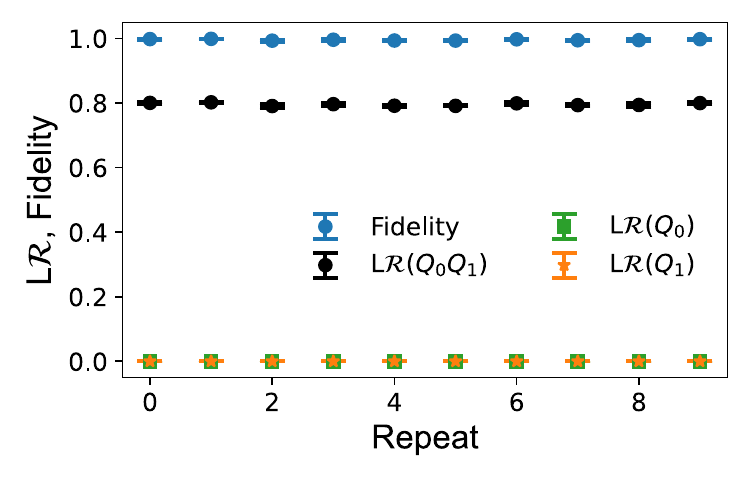}
\caption{Experimental realization of the golden ENM state.
The blue points show the reconstructed-state fidelity for ten repeated data sets.
The black, green, and orange points show the log-RoM of the joint
two-qubit state and its two single-qubit marginals ($Q_0$ and $Q_1$),
respectively.
The global log-RoM approaches the theoretical maximum
$\ln\sqrt{5}\simeq0.805$, whereas both marginal log-RoMs remain near
the stabilizer value of $0$.
Each point is obtained by averaging over five experimental runs.
Error bars indicate statistical uncertainties.}
\label{fig:golden-state-experiment}
\end{figure}

\end{document}